\documentclass[%
amsmath,amssymb,
]{revtex4-1}

\usepackage[paperwidth=210mm,paperheight=297mm,centering,hmargin=2.3cm,vmargin=2.7cm]{geometry}
\usepackage{marvosym}
\usepackage{amsfonts}
\usepackage{amssymb}
\usepackage{amsmath}
\usepackage{amsthm}
\usepackage{dcolumn}
\usepackage{bm,bbm}

\usepackage{braket}

\newcommand{\hr}{{\mathcal H}}

\newcommand{\kr}{{\mathcal K}}

\newcommand{\cc}{{\mathbb C}}
\newcommand{\st}{{\mathcal S}}
\newcommand{\rr}{{\mathbb R}}
\newcommand{\nn}{{\mathbb N}}

\newcommand{\expe}{\mathbbm E}
\newcommand{\eins}{\mathbf 1}
\newcommand{\prob}{\mathrm{Pr}}
\newcommand{\id}{\mathrm{id}}

\newcommand{\tr}{\mathrm{tr}}
\newcommand{\supp}{\mathrm{supp}}

\newcommand{\sr}{\mathrm{sr}}
\newcommand{\bS}{\mathbf S}
\newcommand{\conv}{\mathrm{conv}}
\newcommand{\pr}{\mathfrak{P}}
\newcommand{\fS}{\mathfrak{S}}
\newcommand{\rebd}{\mathrm{rebd}}
\newcommand{\ri}{\mathrm{ri}}
\newcommand{\aff}{\mathrm{aff}}
\newcommand{\cL}{\mathcal{L}}
\newcommand{\cN}{\mathcal{N}}

\newtheorem{theorem}{Theorem}

\newtheorem{corollary}[theorem]{Corollary}

\newtheorem{definition}[theorem]{Definition}
\newtheorem{example}[theorem]{Example}

\newtheorem{lemma}[theorem]{Lemma}

\newtheorem{proposition}[theorem]{Proposition}
\newtheorem{remark}[theorem]{Remark}

\begin{document}
\title{Resource cost results for one-way entanglement distillation and state merging of compound 
and arbitrarily varying quantum sources}
\author{H. Boche}%
\email{boche@tum.de.}
\affiliation{ 
Lehrstuhl f\"ur Theoretische Informationstechnik, Technische Universit\"at M\"unchen, 80290 M\"unchen, Germany
}%
\author{G. Jan\ss en}
\email{gisbert.janssen@tum.de}
\affiliation{ 
Lehrstuhl f\"ur Theoretische Informationstechnik, Technische Universit\"at M\"unchen, 80290 M\"unchen, Germany
}%
\date{\today}

\begin{abstract}
 We consider one-way quantum state merging and entanglement distillation under 
 compound and  arbitrarily varying source models. 
 Regarding quantum compound sources, where the source is memoryless, but the source state an unknown
 member of a certain set of density matrices, we continue investigations begun in the work of Bjelakovi\'{c} \emph{et. al.}
 [Universal quantum state merging, J. Math. Phys. \textbf{54}, 032204 (2013)] and determine the classical as well as entanglement cost of state merging. 
 We further investigate quantum state merging and entanglement distillation protocols for arbitrarily varying 
 quantum sources (AVQS). In the AVQS model, the source state is assumed to vary in an arbitrary manner for 
 each source output due to environmental fluctuations or adversarial manipulation. 
 We determine the one-way entanglement distillation capacity for AVQS, where we invoke the famous 
 robustification and elimination techniques introduced by R. Ahlswede. 
 Regarding quantum state merging for AVQS we show by example, that the robustification and elimination 
 based approach generally leads to suboptimal entanglement as well as classical communication rates.
 \end{abstract}
 \maketitle
\begin{section}{Introduction}
Investigations on communication tasks involving bipartite (or multipartite) sources within the local operations and classical 
communications (LOCC) paradigm made a substantial contribution to the progress in quantum 
Shannon theory which took place over the past two decades. \newline 
Especially the role of shared pure entanglement as a communication resource was clarified and substantiated by 
establishment of LOCC protocols inter-converting shared entanglement with optimal rates.\newline
Two prominent tasks, entanglement distillation and quantum state merging are considered in this work.
Quantum state merging was introduced by Horodecki, Oppenheim, and Winter \cite{horodecki07}. In this setting
a bipartite
quantum source described by a quantum state $\rho_{AB}$ shared by communication parties $A$ (sender) and $B$ 
(receiver) is required to be merged at the receivers site by local operations and classical communication
together with shared pure entanglement as resource, such that in the limit of large blocklengths, the source 
is approximately restored on $B'$s site. The optimal asymptotic net entanglement cost was determined in 
Ref. \citenum{horodecki07} 
to be
$
 S(A|B)
$
ebits of shared entanglement per copy of the state, which was shown to be achievable with optimal classical cost 
$
 I(A;E)
$
bits of $A \rightarrow B$ classical side communication per copy ($I(A;E)$ is the quantum mutual information of
A with an environment E purifying $\rho_{AB}$). This result allows interpretation the negative values of 
$S(A|B)$. 
For states with $S(A|B)$ being negative, quantum state 
merging is possible with net production of shared maximal entanglement which may serve as a credit for 
future quantum communication.\newline
Entanglement distillation is in some sense a task subaltern to quantum state merging, since entanglement 
distillation protocols are readily derived from quantum state merging protocols \cite{horodecki07}. 
In this task, a 
given bipartite quantum source has to be transformed into shared maximal entanglement by LOCC in the 
limit of large number of outputs. The optimal entanglement gain was determined in Ref. \citenum{devetak05c}, 
where a connection to secret key distillation from bipartite
quantum states was exhausted.\newline
However, these results were shown under strong idealizations of the sources. It was assumed, that the 
sources where tasks are performed on, are memoryless and perfectly known. 
Since source uncertainties, may they be present due to hardware 
imperfections of the preparation devices and/or manipulation by adversarial communication parties, are 
inherent to all real-life communication settings, this assumption seems rather restrictive. \newline
The contribution of this work is, to partly drop these conditions. We consider entanglement distillation
and quantum state merging in presence of compound and arbitrarily varying quantum sources.  
A compound memoryless source models a preparation device which emits systems, uncorrelated from output
to output, all described by the same given density matrix, which in turn is not perfectly known to 
the communication
parties, but identified as a member of a certain set $\mathcal{X}$ of quantum states. 
Consequently, the communication parties
have to use protocols which are of sufficient fidelity for each member of the set of states generating 
the compound source. \newline 
In the arbitrarily varying source (AVQS) model, the source state can vary from output 
to output over a generating set of states. This variation can be understood as a natural fluctuation as well 
as a manipulation of an adversarial communication party changing the source state from output to output in an 
arbitrary manner. Consequently, the parties are forced to accomplish the tasks with protocols, 
which are robust in the sense, that work with sufficient fidelity for each possible state sequence. 
In this work, we contribute the following. 
Regarding one-way quantum state merging for compound sources, we 
answer a question left open in the preceding work \cite{bjelakovic13}. We derive protocols which beside
being optimal regarding their entanglement cost, also approximate the lowest classical one-way
communication requirements allowed by corresponding converse theorems\cite{horodecki07,bjelakovic13} 
which lower bound the resource requirements for asymptotically faithful merging schemes. \newline
We use the results on one-way entanglement distillation for compound sources established 
	earlier\cite{bjelakovic13} together with the famous elimination and robustification techniques introduced
	by R. Ahlswede \cite{ahlswede78, ahlswede80} to
	determine the capacity for one-way entanglement distillation from AVQS generated by a set 
	$\mathcal{X}$ of states.
	We show, that the one-way entanglement
	distillation capacity in this case, can be expressed by the capacity function of the compound source 
	generated by the convex hull of the set generating the AVQS.\newline
Considering quantum state merging under the AVQS model, we encounter unexpected behavior.
       Opposite to the intuition gathered by previous results from classical as well as quantum Shannon theory,
       the entanglement as well as classical communication resource costs for one-way merging of an AVQS do 
       not match the costs known for the corresponding compound source generated by the 
       convex hull of $\mathcal{X}$ in general. We demonstrate this fact giving a simple example.

\begin{subsection}{Related Work}
The task of entanglement distillation was subject to several investigations in case of perfectly known 
memoryless quantum sources over the past fifteen years. In this work, we generalize a result from 
Ref. \citenum{devetak05c},
where the entanglement distillation capacity with one-way LOCC for perfectly known
memoryless bipartite quantum sources where determined. Quantum state merging was first considered in
 Ref. \citenum{horodecki07}, where the authors determined the entanglement as well as classical cost of quantum state
merging for the scenario with perfectly known density matrix. Both results where partly generalized to the
case of compound memoryless sources in Ref. \citenum{bjelakovic13} within the one-way LOCC scenario.
In this work we continue and complete considerations made therein by determining the optimal classical cost of 
one-way merging for compound quantum sources.\newline
Communication tasks involving arbitrarily varying channels and sources where considered in classical 
information theory from the late 60's. Here we especially mention the robustification 
\cite{ahlswede80, ahlswede86} and elimination \cite{ahlswede78} techniques developed by Ahlswede in the 70's, which are crucial ingredients of our proof
of the one-way entanglement distillation capacity for AVQS. \newline
Arbitrarily varying channels where also considered in quantum Shannon theory. The first result was 
by Ahlswede and Blinovsky \cite{ahlswede07}, who determined the capacity for transmission of classical messages
over an arbitrarily varying channel with classical input and quantum output. A treatment of arbitrarily 
varying quantum channels was done by Ahlswede, Bjelakovi\'{c}, the first author and N\"otzel  
published in 2013 \cite{ahlswede13}. There, they determined the quantum capacity of an arbitrarily varying quantum channel 
for entanglement transmission, entanglement generation as well as strong subspace transmission.
\end{subsection}

\begin{subsection}{Outline}
We set up the notation used in this paper in Section \ref{sect:notation}, where we also state some conventions
and preliminary facts we use freely in our considerations. The basic concepts relevant for this paper are 
concisely stated and and explained in Section \ref{sect:basic_definitions}. \newline
In Section \ref{comp_merging}, we conclude the investigations on quantum state merging for compound 
sources begun in Ref. \citenum{bjelakovic13}. Explicitly, we show existence of universal one-way LOCCs which 
are 
asymptotically optimal regarding the entanglement as well as classical $A \rightarrow B$ communication cost.
For the proof, we use protocols derived in Ref. \citenum{bjelakovic13}, which are optimal regarding their
entanglement cost but require overmuch classical side communication in some cases. These are refined in a 
sufficient way by combination with an entropy estimating instrument used by the sender, where we utilize 
methods from representation theory of the symmetric groups from Refs. \citenum{keyl01} and 
\citenum{christandl06}. 
Section \ref{av_dist} is devoted to determination of the capacity for 
entanglement distillation from an AVQS under restriction to one-way LOCC. 
We first prove an achievability result in case that the AVQS is generated by a finite set $\mathcal{X}$
of bipartite states.
Here we use entanglement distillation schemes with fidelity going to one exponentially fast 
for the compound source generated by the convex hull of $\mathcal{X}$ from Ref. \citenum{bjelakovic13}, 
together with Ahlswede's robustification and elimination techniques. Afterwards, we extend this 
result to the general case approximating the AVQS generating set by suitable finite AVQS. \newline
We also consider the issue of quantum state merging for AVQS and discover a strange feature of the quantum 
state merging task in this scenario. We show in Section \ref{av_merging}, that in general, the entanglement 
as well as classical 
cost of merging an AVQS generated by a set $\mathcal{X}$ of bipartite state are strictly lower than the costs 
of 
merging the corresponding compound source generated by $\conv(\mathcal{X})$. In Section \ref{conclusion}, we
discuss the results obtained.
\end{subsection}
\end{section}

\begin{section}{Notation and Conventions}\label{sect:notation}
 All Hilbert spaces appearing in this work are considered to be finite dimensional complex vector spaces. 
 $\mathcal{L}(\hr)$ is the set of linear maps and $\st(\hr)$ the set of states (density matrices) on a Hilbert 
 space $\hr$ in our notation. We denote the set of quantum channels, i.e. completely positive and trace 
 preserving (c.p.t.p.) maps from $\mathcal{L}(\hr)$ to $\mathcal{L}(\kr)$ by $\mathcal{C}(\hr, \kr)$ and the set 
 of trace-nonincreasing cp maps by $\mathcal{C}^{\downarrow}(\hr,\kr)$ for two Hilbert spaces $\hr$, $\kr$. 
 \newline
 Regarding states on multiparty systems, we freely make use of the following convention for a system consisting
 of some parties $X,Y,Z$, for instance, we denote $\hr_{XYZ} := \hr_{X} \otimes \hr_Y \otimes \hr_Z$, and denote
 the marginals by the letters assigned to subsystems, i.e. $\sigma_{XZ} := \tr_{\hr_Y}(\sigma)$ for $\sigma \in
 \st(\hr_{XYZ})$ and so on. For a bipartite pure state $\ket{\psi}\bra{\psi}$ on a Hilbert space $\hr_{XY}$,
 we denote its Schmidt rank (i.e. number of nonzero coefficients in the Schmidt representation of $\psi$) by 
 $\sr(\psi)$. We define
 \begin{align}
  F(a,b) := \left\|\sqrt{a}\sqrt{b}\right\|_1^2
 \end{align}
 for any two positive semidefinite operators $a,b$ on $\hr$ (this is the quantum fidelity in case that $a$ and
 $b$ are density matrices). If one of the arguments is a pure state, the fidelity is linear in the 
 remaining argument, explicitly $F$ takes the form of an inner product, 
 \begin{align}
  F(\rho, \ket{\psi}\bra{\psi}) = \braket{\psi, \rho \psi}.
 \end{align}
 Relations between $F$ and the trace distance are well known, we will use the inequalities
 \begin{align}
  F(a,\sigma) \geq \tr(a) - \|a - \rho\|_1 \label{fidelity_norm_1}
 \end{align}
 for a matrix $0 \leq a \leq \eins$ and state $\rho$, and
 \begin{align}
  \|\rho - \sigma\|_1 \leq 2 \sqrt{1 - F(\rho,\sigma)}  \label{fidelity_norm_2}
 \end{align}
  for states $\rho,\sigma$. 
 The von Neumann entropy of a quantum state $\rho$ is 
 defined 
 \begin{align}
  S(\rho) := - \tr(\rho \log \rho),
 \end{align}
 where we denote by $\log(\cdot)$ and $\exp(\cdot)$ the base two logarithms and exponentials throughout this paper.
 Given a quantum state
 $\rho$ on $\hr_{XY}$, we denote the conditional von Neumann entropy of $\rho$ given $Y$ by
 \begin{align}
  S(X|Y,\rho) := S(\rho) - S(\rho_Y),
 \end{align}
 the quantum mutual information by
 \begin{align}
  I(X;Y,\rho) := S(\rho_X) + S(\rho_Y) - S(\rho),
 \end{align}
 and the coherent information by
 \begin{align}
  I_c(X\rangle Y, \rho) := S(\rho_Y) - S(\rho) = - S(X|Y,\rho).
 \end{align}
 A special class of channels mapping bipartite systems, which is of crucial importance for our considerations, 
 are one-way LOCC channels, for which we give a concise definition in the following. For more detailed 
 information, the reader is referred to the appendix on one-way LOCCs given in Ref. \citenum{bjelakovic13} 
 and references therein.
 A quantum instrument $\mathcal{T}$ on a Hilbert space $\hr$ is given  by a set 
 $\{\mathcal{T}_k\}_{k=1}^K \subset \mathcal{C}^{\downarrow}(\hr,\kr)$ of trace non-increasing cp maps, 
 such that $\sum_{k=1}^K \mathcal{T}_k$ is  a channel. In this paper, we will only admit instruments with 
 $|K| < \infty$. With bipartite Hilbert spaces $\hr_{AB}$ and 
 $\kr_{AB}$, a channel $\cN \in \mathcal{C}(\hr_{AB},\kr_{AB})$ is an $A \rightarrow B$ (one-way) LOCC channel, 
 if it is a combination of an instrument $\{\mathcal{T}_k\}_{k=1}^K \subset \mathcal{C}^{\downarrow}
 (\hr_A,\kr_A)$ and a family $\{\mathcal{R}_k\}_{k=1}^K \subset \mathcal{C}(\hr_B,\kr_B)$ of channels in the sense, 
 that it can be written in the form
 \begin{align}
  \mathcal{N}(a) = \sum_{k=1}^K (\mathcal{T}_k \otimes \mathcal{R}_k)(a) && (a \in \mathcal{L}(\hr_{AB})).
  \label{locc_definition}
 \end{align}
 The cardinality of the message set for classical transmission from $A$ to $B$ within the application of 
 $\mathcal{N}$ is $K$ (the number of measurement outcomes of the instrument). \newline
 We denote the set of classical probability distributions on a set $\bS$ by $\pr(\bS)$. The $l$-fold Cartesian
 product of $\bS$ will be denoted $\bS^l$ and $s^l := (s_1,...,s_l)$ will be a notation for elements of 
 $\bS^l$. For positive integer $n$, the shortcut $[n]$ is used to abbreviate the set $\{1,...,n\}$.
 For two probability distributions $p,q \in \pr(\bS)$ on a finite set $\bS$, the relative entropy of
 $p$ with respect to $q$ is defined
 \begin{align}
  D(p||q) 
  := \begin{cases}\sum_{s \in \bS} p(s) \log\frac{p(s)}{q(s)} \hspace{0.3cm} \text{if}\ p \ll q \\  
   \infty  \hspace{0.3cm} \text{else}  \end{cases}
  \end{align}
 where $p \ll q$ means $\forall s \in \bS: q(s) = 0 \Rightarrow p(s) = 0$. We denote the Shannon entropy of a
 probability distribution $p$ by $H(p)$.
 For a set $A$ we denote the convex hull of $A$ by $\conv(A)$. If $\mathcal{X} := \{\rho_s\}_{s \in \bS}$ is 
 a finite set of states on a Hilbert space $\hr$, it holds
 \begin{align}
  \conv(\mathcal{X}) = \left\{\rho_p \in \st(\hr):\ \rho_p = \sum_{s\in \bS} p(s) \ \rho_s, \ q \in \pr(\bS)
    \right\}. \label{conv_hull_def}
 \end{align}
 By $\fS_l$, we denote the group of permutations on $l$ elements, in this way
 $\sigma(s^l) = (s_{\sigma(1)},...,s_{\sigma(l)})$ for each $s^l = (s_1,...,s_l)\in \bS^l$ and permutation $\sigma \in
 \fS_l$.\newline
 For any two nonempty sets $\mathcal{X}$, $\mathcal{X}'$ of states on a Hilbert space $\hr$, the Hausdorff
 distance between $\mathcal{X}$ and $\mathcal{X}'$ (induced by the trace norm $\|\cdot\|_1$) is defined by
 \begin{align}
  d_H(\mathcal{X},\mathcal{X}') 
  := \max \left\{ \sup_{\sigma \in \mathcal{X}}\inf_{\sigma' \in \mathcal{X}'}\|\sigma - \sigma'\|_1, 
  \sup_{\sigma' \in \mathcal{X}'}\inf_{\sigma \in \mathcal{X}}\|\sigma - \sigma'\|_1 \right\}.
 \end{align}
 \end{section}

\begin{section}{Basic Definitions} \label{sect:basic_definitions}
 In this section, we define the underlying scenarios, considered in the rest of this paper. 
 Given any set 
 $\mathcal{X}:= \{\rho_s\}_{s \in \bS}\subset \st(\hr)$ of states on a Hilbert space $\hr$, the 
 \emph{compound source generated by $\mathcal{X}$} (or \emph{the compound source $\mathcal{X}$}, 
 for short) is given by the family $\{\{\rho_s^{\otimes l}\}_{s \in \bS}\}_{l \in \nn}$ of states. 
 The above definition models a memoryless quantum source under uncertainty of the statistical parameters.
 The source outputs each system according to a constant density matrix, while the density matrix 
 itself is not known perfectly by the communication parties. It only can be identified as a member of
 $\mathcal{X}$. \newline
 The \emph{arbitrarily varying quantum  source (AVQS) generated by $\mathcal{X}$} 
 (or \emph{the AVQS $\mathcal{X}$})
 is given by the family  $\{\{\rho_{s^l}\}_{s^l \in S^l}\}_{l \in \nn}$, where we use the definition
 \begin{align}
  \rho_{s^l} := \rho_{s_1} \otimes ... \otimes \rho_{s_l}
 \end{align}
 for each member $s^l = (s_1,...,s_l)$ of $\bS^l$. 
 In the AVQS model, the source density matrix can be chosen from the set $\mathcal{X}$ 
 independently for each output. The variation in the source state models hardware imperfections, where 
 the source is subject to fluctuations in the state on one hand. On the other hand, this definition also
 can be understood as a powerful communication attack, where the statistical parameters of the source are, 
 to some extend, perpetually manipulated by an adversarial communication party.
 \begin{subsection}{Quantum State Merging} \label{state_merging_definitions}
  We first give a concise notion of the protocols we admit for quantum state merging. We are interested in 
  the entanglement as well as classical resource costs of quantum state merging. \newline 
  A quantum channel $\mathcal{M}$ is an \emph{$(l,k_l,D_l)$ $A \rightarrow B$ merging} 
  for bipartite sources on $\hr_{AB} := \hr_{A} \otimes \hr_B$, if it is an $A \rightarrow B$ LOCC channel
  (according to the definition from (\ref{locc_definition}))
  \begin{align}
   \mathcal{M}: 
   \cL(\mathcal{K}_{0,AB}^l \otimes \hr_{AB}^{\otimes l}) \rightarrow 
   \cL(\mathcal{K}_{1,AB}^l \otimes \hr_{B'B}^{\otimes l}),
  \end{align}
   with $k_l := \dim \kr_{A,0}^l / \dim \kr_{A,1}^l$, where we assume $\kr_{A,i} \simeq \kr_{B,i}$ ($i=1,2$),
   and 
   \begin{align}
    \mathcal{M}(x) = \sum_{k=1}^{D_l} \mathcal{A}_k \otimes \mathcal{B}_k(x). &&(x \in \cL(\mathcal{K}_{0,AB}^l 
    \otimes \hr_{AB}^{\otimes l}))
   \end{align}
   where $\{\mathcal{A}_k\}_{k=1}^{D_l} \subset \mathcal{C}^{\downarrow}(\kr_{0,A}^l \otimes 
   \hr_{A}^{\otimes l},\kr_{1,A}^l)$ constitutes an instrument and $\{\mathcal{B}_k\}_{k=1}^{D_l}
   \subset \mathcal{C}(\kr_{B,0}^l \otimes \hr_{B}^{\otimes l}, \kr_{B,1}^l \otimes \hr_{B'B}^{\otimes l})$ 
   is a set of channels depending on the parameter $k \in [D_l]$.  
   The spaces $\kr_{AB,0}^l, \kr_{AB,1}^l$ are understood to represent bipartite systems shared by $A$ and 
   $B$, which carry the input 
   and output entanglement resources used in the process.  
   As a convention, we will incorporate the maximally entangled states $\phi_i^l \in \st(\kr_{AB,i}^l)$, $i = 0,1$ into 
   the definition of the protocol, it holds
   \begin{align}
    k_l :=  \frac{\dim \mathcal{K}^l_{0,A}}{\dim \mathcal{K}^l_{1,A}} = 
	 \frac{\dim \mathcal{K}^l_{0,B}}{\dim \mathcal{K}^l_{1,B}}  
	 = \frac{\sr(\phi_0^l)}{\sr(\phi_1^l)}.
    \end{align}  
   We define the \emph{merging fidelity} of $\mathcal{M}_l$ given a state 
   $\rho^l \in \st(\hr_{AB}^{\otimes l})$ by 
    \begin{align}
     F_m(\rho^l, \mathcal{M}_l) := F\left(\mathcal{M}_l\otimes \id_{\hr_E^l}(\phi_0^l \otimes 
      \psi^l), \phi_1^l \otimes \psi'^l\right). \label{merging_fid_def}
    \end{align}
   Here, $\psi^l$ is a purification of $\rho^l$ with an environmental system described on an additional 
   Hilbert space $\hr_E^l$ (usually $\hr_E^l = \hr_{E}^{\otimes l}$ with some space $\hr_E$), and 
   $\psi'^l$ is a state identical to $\psi^l$ but defined on $\hr_{B'B}^{\otimes l}$ completely under control 
   of $B$. 
   It was shown in Ref. \citenum{bjelakovic13} (Lemma 1), that the r.h.s. of (\ref{merging_fid_def}) does 
   not depend on the chosen purification (which justifies the definition of $F_m$), and that the function 
   $F_m$ is convex in the first and linear in the second argument. For the rest of this section, we assume
   $\mathcal{X}:= \{\rho_s\}_{s \in \bS}$ to be any set of bipartite states on $\hr_{AB}$.
  \begin{definition}\label{comp_merging_achievable_rate}
   A number $R_q \in \rr$ is called an \emph{achievable entanglement cost for $A \rightarrow B$ merging
   of the compound source $\mathcal{X}$ with classical communication rate $R_c$}, if 
   there exists a sequence $\{\mathcal{M}_l\}_{l\in \nn}$ of $(l,k_l,D_l)$ $A\rightarrow B$ mergings, 
   such that the conditions 
   \begin{enumerate}
    \item $\underset{l \rightarrow \infty}{\lim}\ \underset{\rho \in \mathcal{X}}{\inf} 
	   \ F_m(\rho^{\otimes l},\mathcal{M}_l) = 1 $ 
    \item $\underset{l \rightarrow \infty}{\limsup}\ \frac{1}{l}\log k_l \leq R_q$
    \item $\underset{l \rightarrow \infty}{\limsup}\ \frac{1}{l} \log D_l \leq R_c$
   \end{enumerate}
    are satisfied. 
    \end{definition}
    In the following definition, priority lies on the optimal entanglement consumption (or gain) of 
    merging processes, while the classical communication requirements are of subordinate priority. However,
    the classical communication is required to be rate bounded in the asymptotic limit. Since the 
    classical communication requirements
    are of interest as well, we also determine the optimal classical communication cost in Section 
    \ref{comp_merging}. 
    \begin{definition}
    The \emph{$A \rightarrow B$ merging cost} $C_{m,\rightarrow}^{AV}(\mathcal{X})$ 
    of the compound source $\mathcal{X}$ is defined by
  \begin{align}
   C_{m,\rightarrow}(\mathcal{X}) := \inf\left\{R_q \in \rr: \ \begin{array}{l}
	      R_q \ \textrm{is an achievable entanglement cost for}\ A \rightarrow B \\ 
	      \textrm{merging of the compound source} \ \mathcal{X}\ \textrm{with}\\
	      \textrm{some classical communication rate}\ R_c
	  \end{array} \right\}.
  \end{align}	
  \end{definition}
  We recall the following theorem proven in Ref. \citenum{bjelakovic13}
  \begin{theorem}[cf. Ref. \citenum{bjelakovic13}] \label{before_merging_cost}
   \begin{align}
    C_{m,\rightarrow}(\mathcal{X}) = \sup_{\rho \in \mathcal{X}} \ S(A|B,\rho).
   \end{align}

  \end{theorem}

  \begin{definition}\label{avqs_merging_achievable_rate}
   A number $R_q \in \rr$ is called an \emph{achievable entanglement cost for $A \rightarrow B$ merging 
   of the AVQS $\mathcal{X}$ with classical communication rate $R_c$} if there exists a sequence 
   $\{\mathcal{M}_l\}_{l\in \nn}$ of $(l,k_l,D_l)$ $A\rightarrow B$ mergings satisfying 
   \begin{enumerate}
    \item $\underset{l \rightarrow \infty}{\lim}\ \underset{s^l \in \bS^l}{\inf} 
	    F_m(\rho_{s^l},\mathcal{M}_l) = 1 $ 
    \item $\underset{l \rightarrow \infty}{\limsup}\ \frac{1}{l}\log k_l \leq R_q$
    \item $\underset{l \rightarrow \infty}{\limsup}\ \frac{1}{l} \log D_l \leq R_c$.
   \end{enumerate}
   \end{definition}
   \begin{definition}
    The $A \rightarrow B$ merging cost $C_{m,\rightarrow}^{AV}(\mathcal{X})$ of the AVQS $\mathcal{X}$ is 
    defined by
  \begin{align}
   C_{m,\rightarrow}^{AV}(\mathcal{X}) := \inf\left\{R_q \in \rr: \ \begin{array}{l}
	      R_q \ \textrm{is an achievable entanglement cost for } \ A \rightarrow B\ \textrm{merging}\\ 
	      \ \textrm{of the AVQS} \ \mathcal{X}\ \textrm{with some classical communication rate}\  R_c
	 \end{array} \right\}
  \end{align}	
  \end{definition}
 \end{subsection}
   
 \begin{subsection}{Entanglement Distillation}
  Concerning entanglement distillation, we are interested in the asymptotically entanglement gain of one-way 
  LOCC distillation
  procedures. We use the following definitions.
 
  \begin{definition}\label{avqs_distillation_achievable_rates}
   A non-negative number $R$ is an \emph{achievable $A \rightarrow B$ entanglement distillation rate 
   for the AVQS generated by a set $\mathcal{X}$ with classical rate $R_c$}, 
   if there exists a sequence 
   $\{\mathcal{D}_l\}_{l \in \nn}$ of $A \rightarrow B$ LOCC channels,
   \begin{align}
    \mathcal{D}_l = \sum_{m=1}^{M_l} \mathcal{A}_{m,l} \otimes \mathcal{B}_{m,l} && (l \in \nn)
   \end{align}
   such that the conditions
   \begin{enumerate}
    \item $\underset{l \rightarrow \infty}{\lim}\ \underset{s^l \in \bS^l}{\inf} F(\mathcal{D}_l(\rho_{s^l}),\phi_l)
            =1$
    \item $\underset{l \rightarrow \infty}{\liminf}\ \frac{1}{l} \log \sr(\phi_l) \geq R$
    \item $\underset{l \rightarrow \infty}{\limsup} \frac{1}{l} \log M_l \leq R_c$
    \end{enumerate}
   are fulfilled, where $\phi_l$ is a maximally entangled state shared by $A$ and $B$ for each $l \in \nn$.
   \end{definition}
   In this paper, we will be primarily interested in the entanglement gain of one-way entanglement 
   distillation. Regarding the
   classical communication cost of entanglement distillation, no general cost results are known even in case 
   that the source is memoryless with perfectly known source state \cite{devetak05c}.
     \begin{definition}
   The $A \rightarrow B$ entanglement distillation capacity for the AVQS generated by $\mathcal{X}$ is defined 
   \begin{align}
    D^{AV}_{\rightarrow}(\mathcal{X}) := \sup\left\{R :\begin{array}{l} R \ \textrm{is an achievable} \ A \rightarrow B 
			       \ \textrm{entanglement distillation rate for} \\ \textrm{the AVQS 
					      }\ \mathcal{X} \textrm{with some classical communication rate}\ R_c\end{array}\right\}.
   \end{align}
  \end{definition}
  The corresponding definitions for achievable rates and entanglement distillation capacity of compound sources
  can be easily guessed (see Ref. \citenum{bjelakovic13}). To introduce some notation we use in this paper, 
  we state the a theorem from Ref. \citenum{devetak05c}, where the 
  $A \rightarrow B$ entanglement distillation capacity 
  $D_{\rightarrow}(\rho)$ of a memoryless bipartite quantum source with
  perfectly known density matrix $\rho$ was considered. 
  \begin{theorem}[Ref. \citenum{devetak05c}, Theorem 3.4]
   Let $\rho$ be a state on $\hr_{AB}$. It holds
   \begin{align}
    D_{\rightarrow}(\rho) = \lim_{k \rightarrow \infty} \frac{1}{k} \sup_{\mathcal{T} \in \Theta_k} 
    D^{(1)}(\rho^{\otimes k}, \mathcal{T}) \label{single_state_ent_dist}
   \end{align}
    with
   \begin{align}
     D_{\rightarrow}^{(1)}(\sigma, \mathcal{T}) := \sum_{\substack{j \in [J]: \\ \lambda_{j}(\sigma)\neq 0}} 
     \lambda_j(\sigma) \ I_c(A\rangle B, \sigma_j),
     \label{ent_dist_cap_func_oneshot}
   \end{align}
    where $\Theta_k$ is the set of finite-valued quantum instruments on $A$'s site, i.e. 
    \begin{align}
     \Theta_k := \left\{ \{\mathcal{T}_j\}_{j=1}^J \subset \mathcal{C}^{\downarrow}(\hr_A^{\otimes k}, \kr_A): 
	\sum_{j=1}^J \mathcal{T}_j \in \mathcal{C}(\hr_A^{\otimes k}, \kr_A),\ J < \infty, \
	\dim\kr_A < \infty \right\}. \label{instrument_set_defined}
    \end{align}
    For each state $\sigma$ and quantum instrument $\mathcal{T} := \{\mathcal{T}_j\}_{j=1}^J$ on $A$'s site 
    and definitions
    \begin{align}
     \lambda_j(\sigma) := \tr(\mathcal{T}_j(\sigma_A)), \hspace{0.1cm}\text{and} \hspace{0.3cm}
     \sigma_j := \frac{1}{\lambda_j}(\sigma)(\mathcal{T}_j \otimes \id_{\hr_B})(\sigma)
    \end{align}
    for each $j$ with $\lambda_j(\sigma) \neq 0$.
  \end{theorem}
  \begin{remark}
   It is known \cite{devetak05c}, that the limit in (\ref{single_state_ent_dist}) exists for each 
   state, and maximization over instruments in this formula is always realized by an instrument
   $\mathcal{T} = \{\mathcal{T}_j\}_{j=1}^J$ with $J \leq \dim \hr_A^{2k}$ and the operation
   $\mathcal{T}_j$ described by only one Kraus operator for $1 \leq j \leq J$.
   \end{remark}
  In order to obtain a compact notation for the capacity functions arising in the entanglement distillation  
  scenarios we consider in this paper, we introduce a one-way LOCC $\hat{\mathcal{T}} := 
  \sum_{j=1}^J \mathcal{T}_j \otimes \ket{e_j}\bra{e_j}$ for each instrument $\{\mathcal{T}_j\}_{j=1}^J$ 
  with domain $\hr_A$ and an orthonormal system $\{e_j\}_{j=1}^J$ in a suitable space 
  $\hr'_B \simeq \cc^{J}$ assigned to $B$, it holds
  \begin{align}
   D^{(1)}(\sigma,\mathcal{T}) = I_c(A\rangle BB',\hat{\mathcal{T}}(\sigma))\label{ent_dist_cap_func_oneshot_2}
  \end{align}
  in (\ref{ent_dist_cap_func_oneshot}) for each given state $\sigma$. 
  \end{subsection}
 \end{section}
\begin{section}{Quantum State Merging for Compound Quantum Sources}\label{comp_merging}
 In this section, we derive, for any given bipartite compound source $\mathcal{X}$, asymptotically faithful 
 state merging protocols, which are approximately optimal regarding their entanglement as well as 
 classical $A \rightarrow B$ communication cost given the corresponding converse statement \cite{bjelakovic13}.
 While the merging cost was determined in Ref. \citenum{bjelakovic13} before (see Theorem 
 \ref{before_merging_cost} above also), the protocols used there, are 
 suboptimal, in general, regarding their classical $A \rightarrow B$ communication requirements.
 However, it was shown there (see Section V in Ref. \citenum{bjelakovic13}), that 
 \begin{align}
  R_c = \sup_{\rho \in \mathcal{X}} \ I(A;E,\rho)
 \end{align}
 (supremum of the quantum mutual information between $A$ and a purifying environment $E$) is a lower bound
 on the $A \rightarrow B$ classical communication cost for merging a compound source $\mathcal{X}$ by 
 by protocols which have fidelity one in the limit of large blocklengths. \newline
 Proposition \ref{comp_merging_complete} below states, that this bound actually is achievable, and thus 
 together with results from Ref. \citenum{bjelakovic13} provides a full solution of the quantum state merging 
 problem for compound quantum sources.  
 The assertions proved in this section will be utilized in Section \ref{av_merging}, where we compare 
 the $A \rightarrow B$ merging as well as the classical communication cost of a certain AVQS merging protocol 
 for a set $\mathcal{X}$ with the optimal costs of state merging protocols for the compound source generated 
 by $\conv(\mathcal{X})$.\newline
 The preliminary Proposition \ref{pre_comp_merg} below is a slight generalization of Theorem 6 in 
 Ref. \citenum{bjelakovic13}. It states existence of protocols achieving the optimal entanglement cost, but
 with generally suboptimal classical communication rates. However, these protocols will be utilized to 
 derive protocols suitable for the proof of Proposition \ref{comp_merging_complete}.
  \begin{proposition}[cf. Ref. \citenum{bjelakovic13}, Theorem 6] \label{pre_comp_merg}
  Let $\mathcal{X} \subset \st(\hr_{AB})$. 
  For each $\delta > 0$, there is a number $l_0 \in \nn$, such that for each blocklength 
  $l > l_0$ there is an $(l,k_l,D_l)$ $A \rightarrow B$ merging $\mathcal{M}_l$, such that
  \begin{align}
   \inf_{\rho \in \mathcal{X}}\ F(\rho^{\otimes l},\mathcal{M}_l) \geq 1 - 2^{-lc_1}
  \end{align}
  with a a constant $c_1 = c_1(\mathcal{X},\delta) > 0$,  
  \begin{align}
  \frac{1}{l}\log k_l \leq \sup_{\rho \in \mathcal{X}} S(A|B,\rho) + \delta
  \label{comp_merg_prop_pre_1}
  \end{align}
  and
  \begin{align}
   \frac{1}{l} \log D_l \leq  \sup_{\rho \in \mathcal{X}}S(\rho_{A}) + \sup_{\rho \in \mathcal{X}} S(A|B,\rho) + \delta.
  \end{align}
 \end{proposition}
 \begin{proof}
  The assertion to prove includes both, a strengthening of the fidelity convergence rates in Ref. 
  \citenum{bjelakovic13},
  Theorem 4 to exponentially decreasing trade-offs, and a generalization of Theorem 6 in 
  Ref. \citenum{bjelakovic13} to arbitrary (not necessary finite or countable) sets of states.\newline
  Approximating $\mathcal{X}$ by a $\tau_l$-net 
  $\mathcal{X}_{\tau_l} := \{\rho_i\}_{i = 1}^{N_{\tau_l}} \subset \st(\hr_A \otimes \hr_B)$ for each 
  blocklength $l$ (see Ref. \citenum{bjelakovic13} 
  for details) and using the result for finite sets, we infer by careful observation of the merging fidelities 
  in Ref. \citenum{bjelakovic13} (see eqns. (36), (37), and (58) therein), that for given $\delta > 0$ and large enough blocklength $l$, there exists a 
  $(l,k_l,D_l)$ $A\rightarrow B$ merging $\mathcal{M}_l$, where
  \begin{align}
   \inf_{\rho \in \mathcal{X}} \ F_m(\rho^{\otimes l},\mathcal{M}_l) 
   \geq 1 - N_{\tau_l}^2 \cdot 2^{-l \theta} - 4\sqrt{l \cdot \tau_l} \label{pre_comp_merg_prf_1}
  \end{align}
  is valid for the merging fidelities with a constant 
  $\theta = \theta(\delta) > 0$, and
  \begin{align}
   \frac{1}{l} \log k_l \leq \underset{\rho \in \mathcal{X}}{\sup}\ S(A|B,\rho) + \frac{\delta}{2}.
   \label{pre_comp_merg_prf_1_1}
  \end{align}
  (see (57) in Ref. \citenum{bjelakovic13}). Moreover, we can bound the number of messages for the 
  classical
  $A \rightarrow B$-communication (see (101) in Ref. \citenum{bjelakovic13}) by,
  \begin{align}
   \frac{1}{l}\log D_l 
   &\leq \underset{1 \leq i \leq N_{\tau_l}}{\max}\ S(\rho_{A,i})
   + \underset{1 \leq i \leq N_{\tau_l}}{\max}\ S(A|B,\rho_i) + \frac{\delta}{2} + 
   \frac{1}{l}\log N_{\tau_l} \\
   &\leq \underset{\rho \in \mathcal{X}}{\sup}\ S(\rho_{A})
   + \underset{\rho \in \mathcal{X}}{\sup}\ S(A|B,\rho) + \overline{\nu}(\tau_l) + \frac{\delta}{2} + 
   \frac{1}{l}\log N_{\tau_l}, \label{pre_comp_merg_prf_clss}
  \end{align}
  where the summand $\overline{\nu}(\tau_l) := 3 \tau_l \log \frac{\dim\hr_{AB}}{\tau_l}$ follows from 
  threefold application of Fannes' inequality \cite{fannes73},  i.e. 
  \begin{align}
   \left|\underset{1 \leq i \leq N_{\tau_l}}{\max} S(\rho_{A,i})
   + \underset{1 \leq i \leq N_{\tau_l}}{\max} S(A|B,\rho_i) 
   - \underset{\rho \in \mathcal{X}}{\sup}\ S(\rho_{A})
   + \underset{\rho \in \mathcal{X}}{\sup}\ S(A|B,\rho)\right| \leq \overline{\nu}(\tau_l).
  \end{align}
  Due to the bound given in Ref. \citenum{bjelakovic13}, Lemma 9, it is known, that the nets can be chosen 
  with cardinality bounded by
  \begin{align}
   N_{\tau_l} \leq \left(\frac{3}{\tau_l} \right)^{2(\dim \hr_{AB})^2}.
  \end{align}
   for each $l \in \nn$. Choosing net parameter $\tau_l = 2^{-l \theta'}$ with $\theta' := 
   \min\{\theta/8(\dim\hr_{AB})^2, \delta/4\}$ for each $l$, we infer
   \begin{align}
    \inf_{\rho \in \mathcal{X}} \ F_m(\rho^{\otimes l},\mathcal{M}_l) 
   \geq 1 - 2^{-l \frac{\theta}{2}} - 2^{-l\frac{\theta'}{4}}
   \geq 1 - 2^{-l c_1} \label{pre_comp_merg_prf_2}
   \end{align}
   with a constant $c_1 = c_1(\delta) > 0$, and
   \begin{align}
    \frac{1}{l}\log D_l 
      \leq \underset{\rho \in \mathcal{X}}{\sup} \ S(\rho_{A})
   + \underset{\rho \in \mathcal{X}}{\sup} \ S(A|B,\rho) + \delta \label{pre_comp_merg_prf_3}
   \end{align}
   from (\ref{pre_comp_merg_prf_clss}) if $l$ is large enough, to satisfy $\overline{\nu}(\tau_l) 
   \leq \frac{\delta}{4}$. Collecting the bounds in (\ref{pre_comp_merg_prf_1_1}), 
   (\ref{pre_comp_merg_prf_2}), and (\ref{pre_comp_merg_prf_3}), we are done. 
 \end{proof}
 Before we state and prove Proposition \ref{comp_merging_complete}, we collect some results from representation
 theory of the symmetric groups, which we utilize in the proof.\newline 
 We denote by $YF_{d,l}$ 
 the set of young frames with at most $d$ rows and $l$ boxes for $d,l \in \nn$. A young frame $\lambda \in 
 YF_{d,l}$ is determined by a tuple $(\lambda_1,...,\lambda_d)$ of nonnegative integers summing to $l$.
 The box-lengths $\lambda_1,...,\lambda_d$ of $\lambda$ define a probability
 distribution $\overline{\lambda}$ on $[d]$ in a natural way 
 via the definition $\overline{\lambda}(i) := \frac{1}{l}\lambda_i$ 
 for each $1 \leq i \leq d$. 
 To each Young frame $\lambda \in YF_{d,l}$, there is an invariant subspace of $(\cc^d)^{\otimes l}$, 
 and we denote by $P_{\lambda,l}$ the projector onto the subspace belonging to $\lambda$.\newline
 Theorem \ref{spectrum_estimation} below allows, to asymptotically estimate the spectrum of a density operator 
 $\rho$ by projection valued measurements on i.i.d. sequences of the form $\rho^{\otimes l}$, and is an 
 important ingredient of our proof of Proposition \ref{comp_merging_complete}. 
 A variant of the first statement of the theorem was first proven in by Keyl and Werner\cite{keyl01}. 
 The actual bounds stated below are from Ref. \citenum{christandl06}, while the remaining statements of the 
 theorem are well-known facts in group representation theory (Ref. \citenum{christandl06} and references 
 therein are recommended for further information). 
 \begin{theorem}[cf. Refs. \citenum{keyl01} and \citenum{christandl06}]\label{spectrum_estimation}
  The following assertions are valid for each $d, l \in \nn$.
 \begin{enumerate}
  \item For $\lambda \in YF_{d,l}$ and $\rho \in \st(\cc^d)$, it holds
  \begin{align}
    \tr(P_{\lambda,l}\rho^{\otimes l}) \leq (l+1)^{d(d-1)/2} \exp(-l D(\overline{\lambda}||r))
  \end{align}
    where $\overline{\lambda} \in \mathfrak{P}([d])$ is the probability distribution given by the normalized 
    box-lengths of $\lambda$, and $r$ is the probability distribution on $[d]$ induced by the decreasingly 
    ordered spectrum of $\rho$ (with multiplicities of eigenvalues counted). 
    \label{spectrum_estimation_1}
  \item $|YF_{d,l}| \leq (l+1)^d$. \label{spectrum_estimation_2}
  \item For $\lambda, \lambda' \in YF_{d,l}$, it holds $P_{\lambda,l} P_{\lambda',l} = 0$ if 
  $\lambda \neq \lambda'$. \label{spectrum_estimation_3}
 \end{enumerate}
\end{theorem}
 \begin{lemma}[Refs. \citenum{winter99}, \citenum{ogawa07}]\label{tender_operator}
 Let $\tau$, $X$ be matrices with $\tau \geq 0$, $\tr(\tau) \leq 1$, and $0 \leq X \leq 1$, $\epsilon \in 
 (0,1)$. If $\tr(\rho X) \geq 1 - \epsilon$, it holds
 \begin{align}
  \|\sqrt{X}\rho\sqrt{X} - \rho \|_1 \leq 2 \sqrt{\epsilon}.
 \end{align}
\end{lemma}
 The following proposition is the main result of this section.
\begin{proposition}\label{comp_merging_complete}
 Let $\mathcal{X} \subset \st(\hr_{AB})$ be a set of states on $\hr_{AB}$.
 For each $\delta > 0$, there exists a number $l_0 = l_0(\delta)$, such that for each $l > l_0$ there  
 is an $(l,k_l,D_l)$ $A \rightarrow B$ merging $\mathcal{M}_l$ with
 \begin{align}
  \underset{\rho \in \mathcal{X}}{\inf} \ F_m(\rho^{\otimes l}, \mathcal{M}_l) \geq 1 - 2^{-lc_2} 
  \label{comp_merging_complete_3}
 \end{align}
 with a constant $c_2=c_2(\mathcal{X},\delta) > 0$, 
 \begin{align}
  \frac{1}{l} \log k_l \leq \sup_{\rho \in \mathcal{X}} S(A|B, \rho) + \delta, \label{comp_merging_complete_1}
 \end{align}
 and
 \begin{align}
  \frac{1}{l}\log D_l \leq \sup_{\rho \in \mathcal{X}} I(A;E,\rho) + \delta, \label{comp_merging_complete_2}
 \end{align}
 where the quantum mutual information in (\ref{comp_merging_complete_2}) is evaluated on the AE marginal state 
 of any purification $\psi$ of $\rho$ (notice, that the above abuse of notation does not lead to ambiguities,
 since $I(A;E, \tr_{AE}(\psi)) = S(\rho_A) + S(A|B,\rho)$ holds for any purification $\psi$ of $\rho$).
\end{proposition}
\begin{remark}
 Regarding the classical communication cost a quite restrictive converse statement was shown to be 
 valid\cite{bjelakovic13}. Asymptotically faithful one-way state merging schemes demand classical 
 communication at rate 
 \begin{align}
  R_c \geq \sup_{\rho \in \mathcal{X}} I(A;E,\rho)
 \end{align}
 regardless of the entanglement rate achieved, i.e. even investing more entanglement resources (choosing 
 protocols with suboptimal merging rates) does not lead to a reduction of the classical communication cost 
 in a significant way.
\end{remark}

\begin{proof}[Proof of Proposition \ref{comp_merging_complete}]
 One half of the above assertion was already proven (see Ref. \citenum{bjelakovic13} and 
 Proposition \ref{pre_comp_merg} at the beginning of this section). Explicitly, it was 
 shown there, that $A \rightarrow B$ LOCC channels exist for each set of bipartite states, which for 
 sufficiently large blocklengths fulfill the conditions formulated in (\ref{comp_merging_complete_3}) and 
 (\ref{comp_merging_complete_1}). We complete the proof by demonstrating, that also the constraint (\ref{comp_merging_complete_2}) on the 
 classical $A \rightarrow B$ communication rate can be met simultaneously with 
 (\ref{comp_merging_complete_3}) and (\ref{comp_merging_complete_1}) by certain 
 protocols. 
 The strategy of our proof will be as follows. We decompose $\mathcal{X}$ into disjoint subsets 
 $\mathcal{X}_1,...,\mathcal{X}_N$, each containing only states with approximately equal 
 entropy on the $A$-marginal system and combine an entropy estimating instrument on the $A$-system with 
 a suitable merging scheme for each set $\mathcal{X}_i$ according to Proposition
 \ref{pre_comp_merg}. We fix $\delta > 0$,  and assume, to simplify the argument, that 
 \begin{align}
  \tilde{s} := \underset{\rho \in \mathcal{X}}{\sup}\ S(A|B,\rho) < 0 \label{cond_entr_neg}
 \end{align}
 holds (i.e. merging is possible without input entanglement resources for large enough blocklengths). 
 Otherwise the argument below can be carried out using further input entanglement and wasting it 
 before action of the protocol. We define $d := \dim \hr_A$ and fix $\eta \in (0,1]$
 to be determined later. Consider the sequence 
 \begin{align}
 s_0 := 0 < s_1 <...< s_{N} := \log d, \hspace{0.3cm} s_i := s_{i-1} + \eta\ \text{for each}\ 1 \leq i < N.
 \end{align}
 Define Intervals $I_1 := [s_0,s_1]$ and 
 $I_i := (s_{i-1}, s_i]$ for $i = 2,...,N$, which generate a decomposition of $\mathcal{X}$ into disjoint
 sets $\mathcal{X}_1,...,\mathcal{X}_{N}$ by definitions
 \begin{align}
  \mathcal{X}_i := \{\rho \in \mathcal{X}: S(\rho_A) \in I_i\} && (i \in [N]), \label{comp_decomp}
 \end{align}
 and set
 \begin{align}
  \tilde{\mathcal{X}}_i := \bigcup_{j \in n(i)} \mathcal{X}_i (i \in [N])
 \end{align}
 where $n(i)$ is defined $n(i) := \{j \in [N]: |j-i| \leq 1\}$ for all $i$. In order to construct an entropy estimating instrument in the $A$ marginal systems, we define 
 an operation $\mathcal{P}_l^{(i)} \in \mathcal{C}^{\downarrow}(\hr_{AB}^{\otimes l},\hr_{AB}^{\otimes l})$ by 
 \begin{align}
  \mathcal{P}_l^{(i)}(\cdot) := p_{i,l} \otimes \eins_{\hr_B^{\otimes l}}(\cdot)
  p_{i,l}^\ast \otimes \eins_{\hr_B^{\otimes l}} \text{with}  \hspace{0.3cm} \hspace{0.3cm} 
      p_{i,l} := \sum_{\substack{\lambda \in YF_{d,l}:\\ H(\overline{\lambda}) \in I_i}} P_{\lambda,l}
 \end{align}
 for each $i \in [N]$ using the notation from Theorem \ref{spectrum_estimation}. Notice, that $p_1,...,p_N$ 
 form a projection valued measure on $\hr_A^{\otimes l}$ due to Theorem 
 \ref{spectrum_estimation}.\ref{spectrum_estimation_3}. By construction, we have for 
 each state $i \in [N]$, $\rho \in \mathcal{X}_i$,
 \begin{align}
  \sum_{j \in [N]\setminus n(i)} \tr(\mathcal{P}_l^{(j)}(\rho^{\otimes l})) 
    &= \sum_{j \in [N]\setminus n(i)} \tr(p_{i,l} \rho_A^{\otimes l}) \label{entropy_estimation_1}\\
    &= \sum_{\substack{\lambda \in YF_{d,l}:\\|H(\lambda)-S(\rho_A)|\geq \eta}} \tr(P_{\lambda,l}
	\rho_A^{\otimes l}) \label{entropy_estimation_2}\\
    &\leq |YF_{d,l}|\cdot (l+1)^{d(d-1)/2}  \\
      &\times \ \exp\left(-l\left(\underset{r:H(r) \in I_i}{\min}
	  \underset{\substack{\lambda \in YF_{d,l}: \\ |H(\overline{\lambda})-H(r)|\geq \eta}}{\min}
	   D(\overline{\lambda}||r)\right) \right) \label{entropy_estimation_3},
    \end{align}
 where (\ref{entropy_estimation_1}) and (\ref{entropy_estimation_2}) are valid due
 to construction and (\ref{entropy_estimation_3}) follows from Theorem 
 \ref{spectrum_estimation}.\ref{spectrum_estimation_1}. Since the 
 relative entropy term in the exponent on the r.h.s. of (\ref{entropy_estimation_3}) is bounded away from 
 zero for each fixed number $\eta>0$ (consult the appendix of this paper for a proof of this fact), i.e.
  \begin{align}
 \underset{r:H(r) \in I_i}{\min}
	  \underset{\substack{\lambda \in YF_{d,l}: \\ |H(\overline{\lambda})-H(r)|\geq \eta}}{\min}
	   D(\overline{\lambda}||r) \geq 2c_3 && (i \in [N])\label{entropy_estimation_constant_bound}
 \end{align}
with a constant $c_3 = c_3(\eta)> 0$, and the functions outside the exponential term are growing 
polynomially for $l \rightarrow \infty$ (see Theorem \ref{spectrum_estimation}.\ref{spectrum_estimation_2})
, we deduce 
\begin{align}
\sum_{j \in [N]\setminus n(i)} \tr(\mathcal{P}_l^{(j)}(\rho^{\otimes l})) \geq 2^{-lc_3} &&(i \in [N])
\label{entropy_estimation_4}
\end{align}
provided that $l$ is large enough. \newline
Define index sets $J := \{i: \mathcal{X}_i \neq \emptyset\}$ and $\tilde{J} := 
\{i: \tilde{\mathcal{X}}_i \neq \emptyset\}$. We know from Proposition \ref{pre_comp_merg}, 
that for each sufficiently large $l$, we find an $(l,k_l,D_l^{(i)})$ $A \rightarrow B$ merging
$\widetilde{\mathcal{M}}^{(i)}_l$ for each $i \in \tilde{J}$ such that
\begin{align}
 \inf_{\rho \in \tilde{\mathcal{X}}_i}\ F_m(\rho^{\otimes l},\widetilde{\mathcal{M}}_l^{(i)})
 &\geq 1 - 2^{-l\tilde{c}_i} \label{comm_sub_comp_2}
\end{align}
 holds with a constant $\tilde{c}_i > 0$, 
 \begin{align}
    -\frac{1}{l}\log k_l \leq \sup_{\rho \in \tilde{\mathcal{X}}_i} S(A|B,\rho)  + \frac{\delta}{2}
   \label{comm_sub_comp_3}
 \end{align}
 and 
 \begin{align}
  \frac{1}{l} \log D_l^{(i)} &\leq \sup_{\rho \in \tilde{\mathcal{X}}_i} S(\rho_A) + 
	\sup_{\rho \in \tilde{\mathcal{X}}_i} S(A|B,\rho) + \frac{\delta}{2} \label{comm_sub_comp_4}
  \end{align}
 for the classical $A \rightarrow B$ communication rate. By construction of the sets 
 $\tilde{\mathcal{X}}_i$, $i \in \tilde{J}$, it also holds
 \begin{align}
 I(A;E,\rho) \ = \ S(\rho_A) + S(A|B,\rho) \ \geq \ \sup_{\rho \in \tilde{\mathcal{X}}_i} S(\rho_A) 
 - 3 \eta + S(A|B,\rho)
 \end{align}
 for each $\rho \in \tilde{\mathcal{X}}_i$. Taking suprema over the set $\tilde{\mathcal{X}}_i$ on both sides of
 the above inequality in combination with (\ref{comm_sub_comp_4}) leads us to the estimate
\begin{align}
 \frac{1}{l}\log D^{(i)}_l \ \leq \ \sup_{\rho \in \tilde{\mathcal{X}}_i} I(A;E,\rho) + \frac{\delta}{2}
		    + 3 \eta  
                     \ \leq \ \sup_{\rho \in \mathcal{X}} I(A;E,\rho) + \frac{\delta}{2} + 3 \eta
 \label{comm_opt_comp_1}
 \end{align}
 for each $i \in [\tilde{J}]$.  Combining the entropy estimating instrument 
 $\{\mathcal{P}_l^{(j)}\}_{j=1}^N$ with the corresponding merging protocols, we define
 \begin{align}
  \mathcal{M}_l(\cdot) := \sum_{i = 1}^N \widetilde{\mathcal{M}}_l^{(i)} \circ \mathcal{P}_l^{(i)}(\cdot). 
  \label{comp_opt_merg_def}
 \end{align}
 The maps $\widetilde{\mathcal{M}}_l^{(i)}$ are yet undefined for all numbers 
 $i \in [N] \setminus \widetilde{J}$.
 Since they will not be relevant for the fidelity, they may be defined by any trivial local 
 operations, with 
 $D_l^{(i)} = 1$ for $i \in [N] \setminus \widetilde{J}$. Moreover, we assume, that the merging rate of 
 $\mathcal{M}_l^{(i)}$ for each $i$ is stuck to the the worst and each $\mathcal{M}_l^{(i)}$ outputs approximately
 the same maximally entangled resource output state $\phi_l$. We can always achieve this by partial tracing and
 local unitaries, which do not further affect the classical communication rates.\newline
 By inspection of the definition in (\ref{comp_opt_merg_def}) one readily verifies, that $\mathcal{M}_l$ is, 
 in fact, an 
 $(l,k_l,D_l)$ $A \rightarrow B$ merging, with
 \begin{align}
  D_l = \sum_{i \in \tilde{J}} D^{(i)}_l + |N - \widetilde{J}|,
\end{align}
 and therefore, classical communication rate bounded by
 \begin{align}
  \frac{1}{l} \log D_l &= \frac{1}{l}\log \left( \sum_{i \in \tilde{J}} D_l^{(i)} + 
			|N - \widetilde{J}|\right)  \\
		      &\leq \frac{1}{l}\log\left(N \cdot \max_{i \in [N]} D_l^{(i)}\right)	\\
		      &\leq \sup_{\rho \in \mathcal{X}} I(A;E,\rho) + \frac{\delta}{2} + 3 \eta 
		      + \frac{\log N}{l} 
		      \label{comp_opt_15}. 
 \end{align}
 It remains to show, that we achieve achieve merging fidelity one  with $\{\mathcal{M}_l\}_{l \in \nn}$ 
 for each $\rho \in \mathcal{X}$ with exponentially decreasing trade-offs for large enough blocklengths.
 Assume $\rho$ is a member of $\mathcal{X}_i$ for any index $i \in J$. Then, it holds
 \begin{align}
  F_m(\rho^{\otimes l},\mathcal{M}_l) 
  & \geq \sum_{j \in n(i)} F_m(\rho^{\otimes l},\widetilde{\mathcal{M}}_l^{(j)}\circ \mathcal{P}_l^{(j)} 
  	  ) \label{comp_opt_fid_bound_1}\\
  & = \sum_{j \in n(i)} F(\rho^{\otimes l},\widetilde{\mathcal{M}}_l^{(j)}\circ \widetilde{\mathcal{P}}_l^{(i)})
   - \sum_{j \in n(i)} \sum_{\substack{k \in n(i)\\k\neq j}} F(\rho^{\otimes l},
     \widetilde{\mathcal{M}}_l^{(j)}\circ \mathcal{P}_l^{(k)}). \label{comp_opt_fid_bound_2} 
  \end{align}
The inequality above holds, because the merging fidelity is linear in the operation and all summands are 
nonnegative together with the definition of $\mathcal{M}_l$. The equality is by some zero-adding of terms 
an using the definition 
$\widetilde{\mathcal{P}}_l^{(i)} := \sum_{j \in n(i)} \mathcal{P}_l^{(j)}$ together with linearity 
of the merging fidelity in the operation again. We bound the terms in (\ref{comp_opt_fid_bound_2}) separately.
Beginning with the second term, we notice, that the fidelity is homogeneous in its inputs and bounded by one 
for states, it holds
\begin{align}
  F(\widetilde{\mathcal{M}}_l^{(j)}\circ \mathcal{P}_l^{(k)} \otimes \id_{\hr_E^{\otimes n}}(\psi_l), 
 \phi_l \otimes \psi_l') \leq \tr(\mathcal{P}_l^{(k)}(\rho_A^{\otimes n})) \label{comp_opt_fid_bound_4}.
\end{align}
Summing up the bounds in (\ref{comp_opt_fid_bound_4}), rearranging the summands and using the definition of 
$\widetilde{\mathcal{P}}_l^{(i)}$, we obtain the bound
\begin{align}
  \sum_{j \in n(i)} \sum_{\substack{k \in n(i)\\k\neq j}} F_m(\rho^{\otimes l},
    \widetilde{\mathcal{M}}_l^{(j)}\circ \mathcal{P}_l^{(k)}) 
    & \leq \sum_{j \in n(i)} \sum_{\substack{k \in n(i)\\k\neq j}} \tr(\mathcal{P}_l^{(k)}(\rho^{\otimes n})) \\
    & = (|n(i)|-1) \, \tr(\widetilde{\mathcal{P}}_l^{(i)}(\rho^{\otimes l}))  \\
    & \leq |n(i)| - 1. \label{comp_opt_fid_td_0}
\end{align} 
To bound the first terms in (\ref{comp_opt_fid_bound_2}), we use the relation between fidelity and trace norm 
in (\ref{fidelity_norm_1}). It then holds, for each $j \in n(i)$, 
\begin{align}
 F_m(\rho^{\otimes l}, \widetilde{\mathcal{M}}_l^{(j)}\circ \widetilde{\mathcal{P}}_l^{(i)})   
  \ \geq \ \tr(\widetilde{\mathcal{P}}_l^{(i)}(\rho^{\otimes l})) - \|\widetilde{\mathcal{M}}_l^{(j)}\circ \widetilde{\mathcal{P}}_l^{(i)}  
   \otimes \id_{\hr_E^{\otimes l}}(\psi_l) - \phi_l \otimes \psi_l' \|_1. \label{comp_opt_fid_td_1}
\end{align}
For the second term in (\ref{comp_opt_fid_td_1}) it holds by zero adding, triangle inequality and monotonicity 
of the trace norm under action of partial traces
\begin{align}
  \|\widetilde{\mathcal{M}}_l^{(j)}\circ \widetilde{\mathcal{P}}_l^{(i)}  
   \otimes \id_{\hr_E^{\otimes l}}(\psi_l) - \phi_l \otimes \psi_l' \|_1  
  \ \leq \ \|\widetilde{\mathcal{M}}_l^{(j)} \otimes \id_{\hr_{E}^{\otimes l}}(\psi_l), \phi_l \otimes (\psi')^{\otimes l} \|_1 
    + \|\widetilde{\mathcal{P}}_l^{(i)}(\rho^{\otimes l})- \rho^{\otimes l} \|_1. \label{comp_opt_fid_td_2}
 \end{align}
 We further yield the bound 
 \begin{align}
  &\|\widetilde{\mathcal{M}}_l^{(j)} \otimes \id_{\hr_{E}^{\otimes l}}(\psi_l) - \phi_l \otimes \psi'_l \|_1 \\
  & \leq 2 \left(1 - F(\widetilde{\mathcal{M}}_l^{(j)} \otimes \id_{\hr_{E}^{\otimes l}}(\psi_l),
   \phi_l \otimes \psi_l')\right)^{\frac{1}{2}}  \\
  & \leq 2 \cdot 2^{-l\frac{\tilde{c}_i}{2}} \label{comp_opt_fid_td_3}
 \end{align}
 by (\ref{fidelity_norm_2}) together with (\ref{comm_sub_comp_2}), and 
 \begin{align}
  \|\widetilde{\mathcal{P}}_l^{(i)}(\rho^{\otimes l})- \rho^{\otimes l} \|_1 
   \ \leq \ 2 \sqrt{1 - \tr(\widetilde{\mathcal{P}}_l^{(i)}(\rho^{\otimes l}))}  
   \ \leq \ 2 \cdot 2^{-l\frac{c_3}{2}}, \label{comp_opt_fid_td_35}
 \end{align}
 where the first inequality is by Lemma \ref{tender_operator}, and the second inequality is valid due to the 
 bound in (\ref{entropy_estimation_4}) along with the fact, that (because $p_{1,l},....,p_{N,l}$ is a 
 resolution of the identity into pairwise orthogonal projections)
 \begin{align}
  1 - \tr(\widetilde{P}^{(i)}(\rho^{\otimes l})) 
  \ = \ \tr\left(\left(\id_{\hr_{AB}^{\otimes l}}-\widetilde{\mathcal{P}}_l^{(i)}\right)(\rho^{\otimes l})\right)   
  \ = \ \sum_{j \in [N]\setminus n(i)} \tr(\mathcal{P}_l^{(j)}(\rho^{\otimes l}))
 \end{align}
 holds.  We define the constant $c_4$ by $c_4 := \min\{\tilde{c}_1,...,\tilde{c}_N,c_3\}.$
 Combining (\ref{comp_opt_fid_td_1}) with (\ref{comp_opt_fid_td_2})-(\ref{comp_opt_fid_td_35}) leads us to 
 the estimate
 \begin{align}
  F_m(\rho^{\otimes l}\widetilde{\mathcal{M}}_l^{(j)}\circ \widetilde{\mathcal{P}}_l^{(i)})  
   &\geq \tr(\widetilde{\mathcal{P}}_l^{(i)}(\rho^{\otimes l}))  - 4 \cdot 2^{-l \frac{c_4}{2}}  \\
   &\geq 1 - 5 \cdot 2^{-l\frac{c_4}{2}} \label{comp_opt_fid_td_4}
 \end{align}
 for each $j \in n(i)$, where the last of the above inequalities, again is by the bound in 
 (\ref{entropy_estimation_4}). By inserting the bounds given in (\ref{comp_opt_fid_td_0}) and 
 (\ref{comp_opt_fid_td_4}) into (\ref{comp_opt_fid_bound_2}), we
 yield
 \begin{align}
   F_m(\rho^{\otimes l},\mathcal{M}_l)  
  & \geq |n(i)| (1 - 5 \cdot 2^{-l\frac{c_4}{2}}) - (|n(i)|-1)  \\
  & \geq 1 - 5|n(i)|\cdot 2^{-l\frac{c_4}{2}}  \\
  & \geq 1 - 15 \cdot 2^{-l\frac{c_4}{2}} \label{comp_opt_fid_td_5}.
  \end{align}
  If we now choose $\eta$ small enough and assume $l_0$ large enough, to suffice
  \begin{align}
   3 \eta + \frac{\log N}{l_0} \leq \delta,  
  \end{align}
  (\ref{comm_sub_comp_3}), (\ref{comp_opt_15}), and (\ref{comp_opt_fid_td_5}) show, that $\mathcal{M}_l$
  has the desired properties for each $l > l_0$. \newline
  The assertion can be proven for the remaining case $\tilde{s} \geq 0$ by
  considering a compound set $\{\rho \otimes \phi_0 : \rho \in \mathcal{X}\}$ with a maximally
  entangled state $\phi_0$ having Schmidt rank large enough to ensure
  $ 
   \sup_{\rho \in \mathcal{X}} S(A|B,\rho \otimes \phi_0) < 0 
  $ 
  and repeat the argument given above for the first case (note, that $I(A;E,\rho \otimes \phi_0) = 
  I(A;E,\rho)$ holds for each state $\rho \in \mathcal{X}$).
 \end{proof}
 \begin{corollary}
  Asymptotically faithful $A \rightarrow B$-one-way quantum state merging of a compound source 
  $\mathcal{X}$ is possible with (quantum) merging cost
  \begin{align}
   C_{m,\rightarrow}(\mathcal{X}) = \sup_{\rho \in \mathcal{X}} S(A|B,\rho)
  \end{align}
  and classical cost 
  \begin{align}
   R_c(\mathcal{X}) = \sup_{\rho \in \mathcal{X}} I(A;E,\rho)
  \end{align}
  (again with the quantum mutual information evaluated on the $AE$ marginal system of a purification 
  $\psi \in \st(\hr_{ABE})$ of $\rho$ for each $\rho \in \mathcal{X}$). \newline
  Especially, the above lines show, that the merging cost as well as the classical $A \rightarrow B$
  communication cost exhibit regular behavior: If two nonempty sets $\mathcal{X}, \mathcal{X}'$ 
  are near in the Hausdorff distance (see Section \ref{sect:notation} for a definition), 
  the costs will be nearly equal as well.
 \end{corollary}
 \end{section}
\begin{section}{Entanglement Distillation for Arbitrarily Varying Quantum Sources} \label{av_dist}
 In this section, we prove a regularized formula for the one-way entanglement distillation capacity where 
 the source is an AVQS generated by a set $\mathcal{X} \subset \st(\hr_{A} \otimes \hr_B)$. \newline 
 We first prove the achievability part in case that $\mathcal{X}$ is finite, where we derive suitable 
 one-way entanglement distillation protocols for the AVQS $\mathcal{X}$ from entanglement distillation 
 protocols which are universal for the compound source $\conv(\mathcal{X})$ with fidelity approaching one
 exponentially fast. In a second step, we generalize this result allowing $\mathcal{X}$ to be any (not 
 necessarily finite or countable) set on $\hr_A \otimes \hr_B$. To this end, we approximate $\mathcal{X}$ 
 by a polytope (which is known to be the convex hull of a finite set of states), where we utilize methods 
 we borrow from Ref. \citenum{ahlswede13}. First we state some facts concerning the continuity of the 
 one-way entanglement distillation capacity functions. 
  \begin{subsection}{Continuity of Entanglement Distillation Capacities}
  Continuity was shown for the capacity functions appearing in coding theorems of several quantum channel 
  coding scenarios\cite{leung09}, here we state and prove uniform continuity for the entanglement distillation 
  capacity functions.
  \begin{lemma}\label{coherent_cont_set}
   Let $\mathcal{Y},\mathcal{Y}' \subset \st(\hr_X \otimes \hr_Y)$ be two nonempty sets of bipartite states
   with Hausdorff distance $0 \leq d_H(\mathcal{Y},\mathcal{Y}') < \epsilon \leq \frac{1}{2}$. It holds 
   for each $k \in \nn$ and c.p.t.p map $\mathcal{N}$ with domain $\mathcal{L}(\hr_{XY}^{\otimes k})$ 
   \begin{align}
    \left|\inf_{\tau \in \mathcal{Y}}I_c(X\rangle Y, \mathcal{N}(\tau^{\otimes k})) 
     -\inf_{\sigma \in \mathcal{Y}'}I_c(X\rangle Y, \mathcal{N}(\sigma^{\otimes k}))\right| 
     \leq k \nu(\epsilon), \label{coherent_cont_set_1}
   \end{align}
   where the function $\nu$ is defined by $\nu(x):= 4 x \log \dim \hr_X + 2 h(x)$ for $x \in (0 \frac{1}{2})$ 
   and $h$ being the binary entropy $h(x) := -x\log x - (1-x) \log(1-x)$.
   \end{lemma}
  \begin{proof}
  We show this assertion with sets containing only one state defined $\mathcal{Y} := \{\tau\}, \mathcal{Y}' := \{\sigma\}$.
  The general assertion in (\ref{coherent_cont_set_1}) follows directly by definition of the Hausdorff distance. 
  The argument parallels the one given in Ref. \citenum{leung09}, Theorem 6 for continuity of the the entropy 
  exchange for channels. 
  Introduce a state $\gamma_{k,n} := \tau^{\otimes n} \otimes \sigma^{\otimes(k-n)}$ for each 
  $0 \leq n \leq k$. By assumption, it holds
  \begin{align}
  \|\gamma_{k,n-1} - \gamma_{k,n}\|_1 \leq \epsilon  \label{einzelabstand}
  \end{align}
  for each $0 < n \leq k$, which implies, via the Alicki-Fannes inequality \cite{alicki03} for the 
  conditional von Neumann entropy
  \begin{align}
   \left|I_c(X\rangle Y, \mathcal{N}(\gamma_{k,n-1})) - 
   I_c(X\rangle Y, \mathcal{N}(\gamma_{k,n}))\right| \leq \nu(\epsilon) \label{alickiabstand}
  \end{align}
  for each $0 < n \leq k$ by (\ref{einzelabstand}) and monotonicity of the trace 
  distance under action of $\mathcal{N}$. 
  Further, it holds
  \begin{align}
  &\left|I_c(X \rangle Y, \mathcal{N}(\tau^{\otimes k})) - 
  I_c(X \rangle Y, \mathcal{N}(\sigma^{\otimes k}))\right| \\
  &= \left|I_c(X\rangle Y, \mathcal{N}(\gamma_{k,k})) - 
   I_c(X\rangle Y, \mathcal{N}(\gamma_{k,0}))\right|  \\
  & = \left|\sum_{n=1}^k \left(I_c(X\rangle Y, \mathcal{N}(\gamma_{k,n-1})) - 
   I_c(X\rangle Y, \mathcal{N}(\gamma_{k,n}))\right)\right|  \\
  & \leq \sum_{n=1}^k\left|I_c(X\rangle Y, \mathcal{N}(\gamma_{k,n-1})) - 
   I_c(X\rangle Y, \mathcal{N}(\gamma_{k,n}))\right| , \label{alicki_ende}
  \end{align}
  where the first equality above is by definition, and the second by adding some zeros. Estimating each 
  summand in (\ref{alicki_ende}) by (\ref{alickiabstand}) concludes the proof.
  \end{proof}
  \begin{corollary} \label{cont_coroll}
   The one-way entanglement distillation capacity $D_\rightarrow$ for memoryless sources with 
   perfectly known source state in (\ref{single_state_ent_dist}) is a uniformly continuous function 
   (considering the trace distance). Explicitly, it holds for $\rho, \sigma \in \st(\hr_A \otimes \hr_B)$ with
   $\|\rho - \sigma\|_1 < \epsilon \leq \frac{1}{2}$, it holds
   \begin{align}
    |D_{\rightarrow}(\rho) - D_{\rightarrow}(\sigma)| \leq \nu(\epsilon).
   \end{align}
   \end{corollary}
 \end{subsection}
\begin{subsection}{AVQS Generated by Finite Sets}
 In this section, we assume $\mathcal{X}$ to be a finite set of bipartite states. We show, that 
 sequences of one-way entanglement distillation protocols for
 the compound source $\conv(\mathcal{X})$ with fidelity going to one exponentially fast
 can be modified to faithful entanglement distillation schemes for the AVQS $\mathcal{X}$.
 We apply Ahlswede's robustification \cite{ahlswede86} and elimination \cite{ahlswede78} techniques. 
 This method of proof is well-known in classical information theory, and found application also in the 
 quantum setting where it was shown to be a useful approach to determine the entanglement transmission 
 capacity of arbitrarily varying quantum channels (AVQC)\cite{ahlswede13}.
 Proposition \ref{comp_ent_dist_ach_thm} below is a generalization and sharpening of Lemma 12 in Ref. 
 \citenum{bjelakovic13} required for our considerations. It asserts achievability of each rate below
 the one-way entanglement capacity for a compound source generated by a set $\mathcal{Y}$, where we 
 drop the condition of finiteness imposed on $\mathcal{Y}$ in Ref. \citenum{bjelakovic13}, Lemma 12. Moreover,
 we show, that each of these rates is achievable by protocols with fidelity approaching one exponentially 
 fast. 
  \begin{proposition}\label{comp_ent_dist_ach_thm}
  Let $\mathcal{Y} \subset \st(\hr_A \otimes \hr_B)$ be a set of bipartite states. For each $k \in \nn, 
  \delta>0$, there exists a number $l_0 = 
  l_0(k,\delta)$ and a constant $c_5 = c_5(k,\delta, \mathcal{X}) > 0$, such for each $l > l_0$, there exists  
  an $A \rightarrow B$ LOCC $\mathcal{D}_l$ fulfilling 
  \begin{align}
   \inf_{\rho \in \mathcal{X}} \ F(\mathcal{D}_l(\rho^{\otimes l}), \phi_l) \geq 1 - 2^{-lc_5},
  \end{align}
  where $\phi_l$ is a maximally entangled state shared by $A$ and $B$ with
  \begin{align}
   \frac{1}{l} \log\sr(\phi_l) \geq \lim_{k \rightarrow \infty} \frac{1}{k} 
   \sup_{\mathcal{T} \in \Theta_k} \inf_{\rho \in \mathcal{X}} D_{\rightarrow}^{(1)}(\rho^{\otimes k}, \mathcal{T}) - \delta.
  \end{align}
  The function $D_{\rightarrow}^{(1)}$ is defined in (\ref{ent_dist_cap_func_oneshot}), and $\Theta_k$ is 
  defined as in (\ref{instrument_set_defined}) for each $k \in \nn$.
 \end{proposition}
 \begin{proof}
  The line of proof is similar to that of the proofs given for Lemma 12 and Theorem 8 in Ref. \citenum{bjelakovic13} where we replace
  usage of Theorem 4 therein by the sharper and more general result Prop. \ref{pre_comp_merg} proven in Sect. \ref{comp_merging} 
  above. We only briefly indicate the line of proof and restrict ourselves to the case $k = 1$. For other 
  $k$, the argument is nearly the same. For given instrument $\mathcal{T}$ on $A's$ systems, and $\delta > 0$, 
  we apply Proposition \ref{pre_comp_merg} to the set $\{\hat{\mathcal{T}}(\rho)\}_{\rho \in \mathcal{X}}$ 
  (remember the notation introduced in (\ref{ent_dist_cap_func_oneshot_2})). In this way, we find for
  each large enough $l$ an $A \rightarrow B$ LOCC $\mathcal{D}_l$ (incorporating $\hat{\mathcal{T}}$), 
  such that
  \begin{align}
   \inf_{\rho \in \mathcal{X}} F(\mathcal{D}_l(\rho^{\otimes l}), \phi_l) \geq 1 - 2^{-lc_5}
  \end{align}
   holds with a maximally entangled state $\phi_l$ with
   \begin{align}
    \frac{1}{l}\log \sr(\phi_l) 
    &\geq - \sup_{\rho \in \mathcal{X}} S(A|BB',\hat{\mathcal{T}}(\rho)) - \frac{\delta}{2}\\
    &= \inf_{\rho \in \mathcal{X}} I_c(A \rangle BB', \hat{\mathcal{T}}(\rho)) - \frac{\delta}{2}.
   \end{align}
   Since this can be done for each instrument $\mathcal{T}$ on $A$'s site. Maximization over instruments 
   on $A$'s site shows the assertion.
 \end{proof}
 Proposition \ref{comp_ent_dist_ach_thm} allows to drop the finiteness condition on the compound generating
 set in Ref. \citenum{bjelakovic13}, Theorem 8. We obtain the following corollary.
 \begin{corollary}[cf. \citenum{bjelakovic13}, Ref. 8] \label{comp_ent_dist_gen}
  Let $\mathcal{Y} \subset \st(\hr_A \otimes \hr_B)$. 
  \begin{enumerate}
  \item It holds
  \begin{align}
   D_{\rightarrow}(\mathcal{Y}) = \lim_{k \rightarrow \infty} \frac{1}{k} 
   \sup_{\mathcal{T} \in \Theta_k} \inf_{\rho \in \mathcal{Y}} D_{\rightarrow}^{(1)}(\rho^{\otimes k}, 
   \mathcal{T}),
  \label{comp_distill_1}
  \end{align}
  where the set $\Theta_k$ is defined as in (\ref{instrument_set_defined}) for each $k \in \nn$. 
  \item
  The function in (\ref{comp_distill_1}) behaves regular for compound sources 
  in the following sense. If $\mathcal{Y},
  \mathcal{Y}' \subset \st(\hr_{AB})$ are two nonempty sets of bipartite states with 
  $d_H(\mathcal{Y},\mathcal{Y}') < \delta \leq \frac{1}{2}$, it holds
  \begin{align}
   |D_{\rightarrow}(\mathcal{Y}) - D_{\rightarrow}(\mathcal{Y}')| \leq \nu(\delta) \label{cont_dist_coroll_1}
  \end{align}
 \end{enumerate}
 \end{corollary}
 \begin{proof}
    Achievability of the r.h.s. in (\ref{comp_distill_1}) directly follows from Proposition 
    \ref{comp_ent_dist_ach_thm}. For the converse statement, we refer to the proof of Theorem 8 in Ref. 
    \citenum{bjelakovic13} for finite sets of states. The argument given there directly carries over to the 
    general case.    
    It remains to show validity of the inequality in (\ref{cont_dist_coroll_1}). Assume $d_H(\mathcal{Y},\mathcal{Y}') 
    < \delta \leq \frac{1}{2}$. Let $\tau>0$ be an arbitrary but fixed 
    number, and $\mathcal{Q}$ be an instrument with domain $\mathcal{L}(\hr_A^{\otimes l})$, such that
    \begin{align}
    \inf_{\rho \in \mathcal{Y}} I_c(A\rangle BB',\hat{\mathcal{Q}}(\rho^{\otimes l})) 
     \geq \sup_{\mathcal{T}\in \Theta_k}\inf_{\rho \in \mathcal{Y}} I_c(A\rangle BB',\hat{\mathcal{T}}(\rho^{\otimes l}))
     - \tau \label{cont_coroll_2}
    \end{align}
    holds, where we used our notation from (\ref{ent_dist_cap_func_oneshot_2}). Lemma \ref{coherent_cont_set}
    implies 
    \begin{align}
     \inf_{\rho \in \mathcal{Y}} I_c(A\rangle BB',\hat{\mathcal{Q}}(\rho^{\otimes l}))
     \geq \inf_{\rho \in \mathcal{Y}'} I_c(A\rangle BB',\hat{\mathcal{Q}}(\rho^{\otimes l})) - k\nu(\delta),
    \end{align}
    which, together with (\ref{cont_coroll_2}) implies
    \begin{align}
     \sup_{\mathcal{T}\in \Theta_k}\inf_{\rho \in \mathcal{Y}'} I_c(A\rangle BB',\hat{\mathcal{T}}(\rho^{\otimes l}))
     \geq \sup_{\mathcal{T}}\inf_{\rho \in \mathcal{Y}} I_c(A\rangle BB',\hat{\mathcal{T}}(\rho^{\otimes l}))
      - \tau - k\nu(\delta).
     \end{align} 
     Since the above line of reasoning also holds with $\mathcal{Y},\mathcal{Y}'$ interchanged and $\tau$ 
     can be chosen arbitrarily small, we obtain
     \begin{align}
     \left|\sup_{\mathcal{T}\in \Theta_k} \inf_{\rho \in \mathcal{Y}}  D^{(1)}(\rho^{\otimes k}, \mathcal{T})
     - \sup_{\mathcal{T}\in \theta_k} \inf_{\rho \in \mathcal{Y}'}  D^{(1)}(\rho^{\otimes k}, \mathcal{T}) \right| 
     \leq k\nu(\delta). 
    \end{align}The above inequality together with the first assertion of the corollary proves the second one.
 \end{proof}
 The following theorem is the core of the robustification technique. It was first proven in 
 Ref. \citenum{ahlswede80}.
 The version below (with a better constant) is from Ref. \citenum{ahlswede86}. 
 \begin{theorem}[Robustification technique, cf. Theorem 6 in Ref. \citenum{ahlswede86}]\label{robustification-technique}\ \\
Let $\bS$ be a set with $|\bS|<\infty$ and $l\in\nn$. If a function $f:\bS^l\to [0,1]$ satisfies
\begin{equation}\label{eq:robustification-1}
 \sum_{s^l\in\bS^l}f(s^l)q(s_1)\cdot\ldots\cdot q(s_l)\geq 1-\gamma
\end{equation}
for each type $q$ of sequences in $\bS^l$ for some $\gamma\in [0,1]$, then
\begin{equation}\label{eq:robustification-2}
  \frac{1}{l!}\sum_{\sigma\in \fS_l}f(\sigma(s^l))\ge 1-(l+1)^{|\bS  |}\cdot \gamma\qquad \forall s^l\in \bS^l.
\end{equation}
\end{theorem}
 The following theorem is the main result of this section.
\begin{theorem}\label{av_dist_ach_fin}
 Let $\mathcal{X}:= \{\rho_s\}_{s \in \bS} \subset \st(\hr_A \otimes \hr_B), |\bS|\leq \infty$. 
 For the AVQS generated by $\mathcal{X}$, it holds
 \begin{align}
  D_{\rightarrow}^{AV}(\mathcal{X}) 
  &\geq D_{\rightarrow}\left(\conv(\mathcal{X})\right) = \lim_{k \rightarrow \infty} 
  \sup_{\mathcal{T}\in \Theta_k} \inf_{p \in \pr(\mathbf{S})} D^{(1)}(\mathcal{T},\rho_p^{\otimes k}),
  \label{av_dist_ach_fin_1}
 \end{align}
 where we use the definition
 \begin{align}
  \rho_{p} := \sum_{s \in \mathbf{S}}\ p^l(s^l) \ \rho_{s^l} \label{av_dist_ach_fin_2}
 \end{align}
 for each $p \in \pr(\mathbf{S})$.
\end{theorem}
\begin{remark} 
 The above statement actually holds with equality in (\ref{av_dist_ach_fin_1}) which we show in the proof
 of Corollary \ref{dist_gen_corr} below.
\end{remark}
\begin{proof}
 We show, that each rate $R$, which is achievable for $A \rightarrow B$ entanglement distillation for 
 the compound source generated by $\conv(\mathcal{X})$, is also an achievable rate for $A \rightarrow B$ 
 entanglement distillation for the AVQS 
 generated by $\mathcal{X}$.
 We indicate the elements of $\conv(\mathcal{X})$ by probability distributions on $\bS$, 
 since 
 \begin{align}
  \conv(\mathcal{X}) = \left\{\rho_p:\ \rho_p = \sum_{s \in \bS} p(s) \rho_s, \ p 
  \in\pr(\bS)\right\} \label{conv_hull_dist}
 \end{align}
 holds.
  We know from Proposition \ref{comp_ent_dist_ach_thm}, that for an achievable $A \rightarrow B$ 
  entanglement distillation rate $R$ for the compound source generated by 
 $\conv(\mathcal{X})$, $\delta > 0$ and each sufficiently large 
 blocklength $l$, there exists a one-way LOCC channel $\tilde{\mathcal{D}}_l$, such that the condition
 \begin{align}
  \underset{p \in \pr(\bS)}{\min} F(\tilde{\mathcal{D}}_l(\rho_p^{\otimes l}),\phi_l) \geq 1 - 2^{-lc_5} 
  \label{av_dist_ach_fin_prf_1}
 \end{align}
 is fulfilled with a maximally entangled state $\phi_l$ shared by $A$ and $B$, 
 such that
 \begin{align}
  \frac{1}{l}\log \sr(\phi_l) \geq R - \delta \label{av_dist_ach_fin_prf_2}
 \end{align}
 holds. Note that the minimization in (\ref{av_dist_ach_fin_prf_1}) is because of (\ref{conv_hull_dist}). 
 We define a function $f: \bS^l \rightarrow [0,1]$ by $f(s^l) := F(\tilde{\mathcal{D}}_l(\rho_{s^l}),\phi_l)$ 
 for each $s^l \in \bS^l$, and infer from (\ref{av_dist_ach_fin_prf_1}), that
 \begin{align}
  \sum_{s^l \in \bS^l} p(s_1)\cdot...\cdot p(s_l) \ f(s^l)\  \geq 1 - 2^{-lc_5}
 \end{align}
 holds for each $p \in \pr(\bS)$ with a constant $c_5 > 0$. Let
 \begin{align}
  \mathcal{U}_{\sigma}(\cdot) := U_{A,\sigma} \otimes U_{B,\sigma}(\cdot)U_{A,\sigma}^\ast 
				\otimes U_{B,\sigma}^\ast,
 \end{align}
 for each permutation $\sigma \in \fS_l$, be the unitary channel, which permutes the tensor factors in 
 $\hr_{AB}^{\otimes l}$ according to $\sigma$, 
 (with unitary matrices $U_{A,\sigma}$, $U_{B,\sigma}$ permuting the tensor bases on $\hr_A^{\otimes l}$ resp.
 $\hr_B^{\otimes l}$). It holds
 \begin{align}
  \rho_{\sigma(s^l)} = \mathcal{U}_{\sigma}(\rho_{s^l}),
 \end{align}
 and consequently
 \begin{align}
  f(\sigma(s^l)) 
   = F(\tilde{\mathcal{D}}_l\circ \mathcal{U}_{\sigma}(\rho_{s^l}),\phi_l) \label{rob_func}
 \end{align}
 for each $s^l \in \bS^l, \sigma \in \fS_l$. The functions in (\ref{rob_func}) fulfill the conditions of 
 Theorem \ref{robustification-technique}, which in turn implies, that
 \begin{align}
  (1 - (l+1)^{|\bS|})\cdot 2^{-lc_5} 
  & \leq \frac{1}{l!}\sum_{\sigma\in \fS_l} F(\tilde{\mathcal{D}}_l \circ 
  \mathcal{U}_\sigma(\rho_{s^l}),\phi_l) \\
  & = F(\hat{\mathcal{D}}_l(\rho_{s^l}),\phi_l)
 \end{align}
 is valid with the definition $\hat{\mathcal{D}}_l := 
 \frac{1}{l!}\sum_{\sigma \in \fS_l} \tilde{\mathcal{D}}_l \circ \mathcal{U}_\sigma$. Notice,
 that $\hat{\mathcal{D}}_l$ is an $A \rightarrow B$ LOCC channel either. However, $\hat{\mathcal{D}}_l$ 
 is not a reasonable protocol for entanglement distillation regarding the classical communication cost.
 Implementation of $\hat{\mathcal{D}}_l$ demands $A \rightarrow B$
 communication of a number of classical messages increased by a factor $l!$ compared to the requirements 
 of $\tilde{\mathcal{D}}_l$,
 which leads to super-exponential growth of required classical messages and consequently unbounded classical 
 communication rates. We remark here, that for a 
 coordination of the permutations in $\hat{\mathcal{D}}_l$, common randomness accessible to $A$ and $B$, 
 which is known to be a weaker resource than $A \rightarrow B$ communication, would suffice. Nonetheless, 
 the asymptotic common randomness consumption of the protocol would be above any rate either. We will a 
 apply the well-known derandomization technique which first appeared in Ref. \citenum{ahlswede78} to 
 construct $A \rightarrow B$ LOCC channel with reasonable classical communication requirements (actually, 
 we will show, that we can approximate the classical cost of $A \rightarrow B$ distillation of the compound 
 source $\conv(\mathcal{X})$. \newline
 Let $X_1,...,X_{K_l}$ be a sequence of i.i.d. random variables, each distributed uniformly on $\fS_l$. 
 We define a function $g: \fS_l \times \bS^l \rightarrow [0,1]$ by
 \begin{align}
  g(\sigma,s^l) = 1 - F(\tilde{\mathcal{D}}_l\circ\mathcal{U}_\sigma (\rho_{s^l}),\phi_l) && 
  (\sigma \in \fS_l , s^l \in \bS^l).
 \end{align}
 One readily verifies, that
 \begin{align}
  \expe\left[g(X_1,s^l)\right] 
  = 1 - F(\hat{\mathcal{D}}_l(\rho_{s^l}),\phi_l) 
  \leq (l+1)^{|\bS|} \ 2^{-lc_5} 
  := \epsilon_l \label{derand_before}
 \end{align}
  holds for each $s^l \in \bS^l$. Thus, for each $s^l \in \bS^l$, and $\nu_l \in (0,1)$, we yield 
  \begin{align}
   \prob\left(\sum_{k=1}^{K_l} g(X_k,s^l) > K_l \nu_l \right) 
   & = \prob\left(\prod_{k=1}^{K_l} \exp(g(X_k,s^l)) > 2^{K_l \nu_l} \right) \label{derand_0}\\
   & \leq 2^{-K_l\nu_l} \cdot \expe\left[\exp(g(X_k,s^l)\right]^{K_l} \label{derand_1}\\
   & \leq 2^{-K_l\nu_l} \cdot  (1 + \expe\left[\exp(g(X_k,s^l)\right])^{K_l} \label{derand_2}\\
   & \leq 2^{-K_l\nu_l} \cdot 2^{K_l\log(1 + \epsilon_l)} \label{derand_3}\\
   & \leq 2^{-K_l(\nu_l-2\epsilon_l)}. \label{derand_4}
  \end{align}
  Eq. (\ref{derand_1}) above is by Markov's inequality, (\ref{derand_2}) follows from the fact, that 
  $\exp(x)\leq 1+x$ holds for $x \in [0,1]$, (\ref{derand_3}) is by (\ref{derand_before}), and (\ref{derand_4}) 
  follows from the inequality $\log(1+x) \leq 2x$ being valid for $x \in (0,1)$. 
  From (\ref{derand_0})-(\ref{derand_4}) and application of de Morgan's laws, it follows
  \begin{align}
   \prob\left(\forall s^l \in \bS^l: \frac{1}{K_l}\sum_{k=1}^{K_l} g(X_k,s^l) \leq \nu_l\right) 
   & \geq 1 - |\bS|^l \cdot 2^{-K_l(\nu_l - 2 \epsilon_l)} \label{av_dist_ach_fin_prf_derand_0} \\
   &\geq 1 - 2^{-l\frac{(\theta - \kappa)}{2}}, \label{av_dist_ach_fin_prf_derand}
  \end{align}
  for large enough $l$,where the last line results from the choosing $\nu_l = 2^{-l\kappa}$ and 
  $K_l = 2^{l\theta}$ with $\theta, \kappa > 0$. 
  If we choose $\kappa$ and $\theta$ in a way, that they fulfill $0 < \kappa < \theta < c$,
  the r.h.s. of (\ref{av_dist_ach_fin_prf_derand}) is strictly positive and we find a
  realization $\sigma_1,...,\sigma_{K_l}$ of $X_1,...,X_{K_l}$, such that for each $s^l \in \bS^l$
  \begin{align}
   2^{-l\kappa}
   &\geq \frac{1}{K_l} \sum_{k=1}^{K_l} g(\sigma_k,s^l)  \\
   & = 1 - \frac{1}{K_l}
	\sum_{k=1}^{K_l}F(\tilde{\mathcal{D}}_l\circ \mathcal{U}_{\sigma_k}(\rho_{s^l}),\phi_l)  \\
   & = 1 - F(\mathcal{D}_l(\rho_{s^l}),\phi_l) \label{av_dist_ach_fin_prf_derand_2},
  \end{align}
  where we defined $\mathcal{D}_l := \frac{1}{K_l}\sum_{k=1}^{K_l}\tilde{\mathcal{D}_l}\circ 
  \mathcal{U}_{\sigma_k}$.
  With (\ref{av_dist_ach_fin_prf_2}) and (\ref{av_dist_ach_fin_prf_derand_2}), it is shown, that for each 
  sufficiently large blocklength $l$, we find a one-way entanglement distillation protocol with
  \begin{align}
   \min_{s^l \in \bS^l} \ F(\mathcal{D}_l(\rho_{s^l}),\phi_l) \geq 1 -2^{-l \kappa}, \hspace{.3cm} \textrm{and}  \hspace{.7cm}
   \frac{1}{l}\log \sr(\phi_l) \geq R - \delta.
  \end{align}
  Notice, that the number of different classical messages to be communicated by $A$ within application of  
  $\mathcal{D}_l$ is increased by a factor $2^{l\theta}$ compared to the message transmission demanded 
  by $\tilde{\mathcal{D}}_l$, i.e. the communication rate
  is increased by $\theta$ (which we can choose to be an arbitrarily small fixed number).   
  \end{proof}
 \end{subsection}
 \begin{subsection}{General AVQS}
 In this section, we generalize the results of the preceding section, admitting the AVQS to be generated by 
 any not necessarily finite or countable set $\mathcal{X}$ of states on $\hr_A \otimes \hr_B$.
 We approximate the closed convex hull of  $\mathcal{X}$ by 
 a polytope, which is known as the convex hull of a finite set of points and apply Theorem 
 \ref{av_dist_ach_fin}, together with continuity properties of the capacity function.
 The proof strategy has some similarities with the argument given in Ref. \citenum{ahlswede13} for entanglement 
 transmission over general arbitrarily varying quantum channels. 
 To prepare ourselves for the approximation, we need some notation and results from convex geometry which we 
 state first.
 For a subset $A$ of a normed space $(V,\|\cdot\|)$, $\overline{A}$ is the closure and $\aff A$ is the affine 
 hull of $A$. If $A$ is a 
 convex set, the relative interior $\ri A$ is the interior and the relative boundary $\rebd A$ of $A$ 
 are the interior and boundary of $A$ regarding the topology on $\aff A$ induced by $\|\cdot\|$. 
  \begin{lemma}[Ref. \citenum{ahlswede13}, Lemma 34]\label{haussdorff_1}
  Let $A$, $B$ be compact sets in  $\cc^n$ with $A \subset B$ and
  \begin{align}
   d_H(\rebd B, A) = t > 0,
  \end{align}
  where $\| \cdot \|$ denotes any norm on $\cc^n$. Let $P$ a polytope with $A \subset P$ and $d_H(A,P) \leq 
  \delta$, where $\delta \in (0,t]$ and $d_H$ is the Hausdorff distance induced by $\|\cdot\|$. Then 
  $P' := P \cap \aff A$ is also a polytope and $P \subset B$. 
 \end{lemma}
 With the above statement and the assertions of the preceding section, we are prepared to prove the following 
 theorem which is the main result of this section.
 \begin{theorem}\label{av_dist_ach_gen}
 Let $\mathcal{X} := \{\rho_s\}_{s \in \bS}$ be a set of states on $\hr_{A} \otimes \hr_B$. For each $\delta
 >0$ and $k \in \nn$, there exists a number $l_0 \in \nn$, such that for each $l > l_0$, there is an 
 $A \rightarrow B$ LOCC channel $\mathcal{D}_l$ fulfilling
 \begin{align}
  \inf_{s\in \bS^l}\ F(\mathcal{D}_l(\rho_{s^l}), \phi_l) \geq 1 - 2^{-lc_6} \label{av_dist_ach_gen_1}
 \end{align}
 with a maximally entangled state $\phi_l$ shared by $A$ and $B$ and a constant $c_6 > 0$, such that 
 \begin{align}
  \frac{1}{l} \log \sr(\phi_l) \geq \frac{1}{k} \sup_{\mathcal{T}\in \Theta_k} \
		\inf_{\tau \in \conv(\mathcal{X})}D_{\rightarrow}^{(1)}(\tau^{\otimes k},\mathcal{T})-\delta
 \end{align}
 holds, where the function $D_{\rightarrow}^{(1)}$ is defined in (\ref{ent_dist_cap_func_oneshot}).
 \end{theorem} 
 \begin{proof}
  Let $\mathcal{T} := \{\mathcal{T}_j\}_{j = 1}^J$ be any instrument with domain 
  $\mathcal{L}(\hr_A^{\otimes k})$, $\delta > 0$. 
  Dealing only with the nontrivial case, we show, that
  \begin{align}
   \inf_{\rho \in \conv(\mathcal{X})}\frac{1}{k}I_c(A\rangle BB', \hat{\mathcal{T}}(\rho^{\otimes k})) 
   - \delta > 0
  \end{align}
  is an achievable rate (remember our notation from (\ref{ent_dist_cap_func_oneshot_2})). 
  Since the Hausdorff distance between $\conv(\mathcal{X})$ and 
  $\overline{\conv(\mathcal{X})}$ is zero, it makes no difference if we consider the 
  set $\overline{\conv(\mathcal{X})}$ instead.
  We briefly describe the strategy of our proof. We approximate the set 
  $\overline{\conv(\mathcal{X})}$
  from the outside by a polytope $P_\eta$. Since $P_\eta$, as a polytope, is the convex hull of a finite set of
  points, Theorem \ref{av_dist_ach_fin} can be applied. 
  A technical issue (cf. Ref. \citenum{ahlswede13}) is, to ensure, that the approximating polytope 
   completely consists of density matrices,
  i.e. $P_\eta \subset \st(\hr_{AB})$. We achieve this by a slight depolarization of the states in
  $\conv(\mathcal{X})$, such that the resulting set does not touch the boundary of $\st(\hr_{AB})$. 
  Define, for $\gamma \in [0,1]$ the channel $\mathcal{N}_\gamma \in \mathcal{C}(\hr_A \otimes \hr_B)$ by
  $\mathcal{N}_\gamma := \mathcal{N}_{A,\gamma} \otimes \mathcal{N}_{B,\gamma}$, where $\mathcal{N}_{X,\gamma}$
  is the $\gamma$-depolarizing channel on the subsystem $X$, $X = A,B$ , i.e
  \begin{align}
   \mathcal{N}_\gamma(\tau) = (1-\gamma)^2 \tau + \gamma(1-\gamma) (\tau_A \otimes \pi_B + \pi_A \otimes
	  \tau_B) + \gamma^2 (\pi_A \otimes \pi_B)
  \end{align}
  for each $\tau \in \st(\hr_A \otimes \hr_B)$, were $\pi_A, \pi_B$ are maximally mixed states
  and $\tau_A, \tau_B$ are the marginals of $\tau$ on $\hr_A$, $\hr_B$. Notice, that $\mathcal{N}_\gamma$ 
  is defined in terms of local depolarizing channels on the subsystems. This is required, since we are 
  restricted to one-way LOCC channels. It holds
  \begin{align}
   \|\mathcal{N}_\eta(\tau) - \tau \|_1 
   &\leq \|(1-\eta)^2 \tau - \tau\|_1 + \eta(1-\eta)\|\tau_A \otimes \pi_B + \pi_A \otimes \tau_B \|_1 \\
   & + \eta \|\pi_A\otimes \pi_B \|_1  \\
   &\leq 6 \eta \label{depol_bound}
  \end{align}
  for each state $\tau$ on $\hr_A \otimes \hr_B$. Moreover, it holds 
  $\overline{\mathcal{N}_\eta(\conv(\mathcal{X}))} = 
  \mathcal{N}_\eta(\overline{\conv(\mathcal{X}})) 
  \subset \ri \st(\hr_A \otimes \hr_B)$, which implies
  \begin{align}
   \inf\left\{\|\rho - \rho'\|_1: \ \rho \in \overline{\mathcal{N}_\eta(\conv(\mathcal{X}))}, 
   \rho' \in \rebd(\st(\hr_A \otimes \hr_B))\right\} > 0.
  \end{align}
  Therefore, due to of Lemma \ref{haussdorff_1} and Theorem 3.1.6 in Ref. \citenum{webster94},  
  there exists, for each small enough
  number $\eta > 0$, a polytope $P_\eta := \conv(\{\tau_e\}_{e \in E_\eta}) \subset \st(\hr_A \otimes \hr_B)$ 
  such that $\mathcal{N}_\eta(\conv(\mathcal{X})) \subset P_\eta$ and
  \begin{align}
   d_H(\mathcal{N}_\eta(\conv(\mathcal{X})), P_\eta) \leq  \eta.
  \end{align}
  Applying Theorem \ref{av_dist_ach_fin} to the finite AVQS generated by the extremal set $\{\tau_e\}
  _{e \in E}$ of the polytope $P_{\eta}$,
  we know, that for each sufficiently large blocklength $l$, there exists an $A \rightarrow
  B$ LOCC channel $\hat{\mathcal{D}}_l$ such that 
  \begin{align}
   F(\hat{\mathcal{D}}_l(\tau_{e^l}), \phi_l) \geq 1 - 2^{-lc_6} \label{av_dist_ach_gen_pol_fid}
  \end{align}
   holds with a maximally entangled state $\phi_l$ shared by $A$ and $B$ for each $e^l \in E^l$ with 
   Schmidt rank fulfilling 
   \begin{align}
    \frac{1}{l}\log \sr(\phi_l) \geq \frac{1}{k} 
     \inf_{\tau \in P_\eta} I_c(A\rangle BB',\hat{\mathcal{T}}(\tau^{\otimes k})) - \frac{\delta}{2} .
     \label{av_dist_ach_gen_pre_rate}
   \end{align}
   Since $\mathcal{N}_\eta(\conv(\mathcal{X})) \subset P_\eta$ holds, the depolarized version 
   $\mathcal{N}_\eta(\rho_s)$ of 
   each state $\rho_s$, $s \in \bS$ can be written as a convex combination of elements 
   from $\{\tau_e\}_{e \in E_\eta}$, i.e. 
   \begin{align}
   \mathcal{N}_{\eta}(\rho_s) = \sum_{e \in E_\eta} q(e|s) \ \tau_e \label{polytope_decomposition}
   \end{align}
   with a probability distribution $q(\cdot|s)$ on $E_\eta$ for each $s \in \bS$. 
   We define a one-way 
   LOCC channel $\mathcal{D}_l$ by $\mathcal{D}_l := \hat{\mathcal{D}}_l \circ \mathcal{N}_\eta^{\otimes l}$ 
   and deduce
   \begin{align}
    F(\mathcal{D}_l(\rho_{s^l}), \phi_l) 
    & = F(\hat{\mathcal{D}}_l(\mathcal{N}_\eta^{\otimes l}(\rho_{s^l})), \phi_l)  \\
    & = F\left(\hat{\mathcal{D}}_l\left(\bigotimes_{i=1}^l \mathcal{N}_\eta(\rho_{s_i})\right), 
    \phi_l\right)\\
    & = F\left(\hat{\mathcal{D}}_l\left(\bigotimes_{i=1}^l \sum_{e_i \in E} q(e_i|s_i) 
    \tau_{e_i}\right),
    \phi_l\right) \label{polytope_decomposition_used}\\
    & = \sum_{e_1 \in E}\cdots \sum_{e_l \in E} \ \prod_{i=1}^l p(e_i|s_i) F\left(\hat{\mathcal{D}}_l
    (\tau_{e_i}),
    \phi_l\right) \\
    & =\sum_{e^l \in E_\eta^l} q^l(e^l|s^l)\ F(\hat{\mathcal{D}}_l(\tau_{e^l}),\phi_l)  \\
    & \geq 1 - 2^{-lc_6} \label{av_dist_ach_gen_fidelity}
   \end{align}
   for each $s^l = (s_1,...,s_l) \in \bS^l$ where we used (\ref{polytope_decomposition}) in 
   (\ref{polytope_decomposition_used}) and (\ref{av_dist_ach_gen_fidelity}) is by 
   (\ref{av_dist_ach_gen_pol_fid}). 
   To complete the proof, we show, that for small enough $\eta$, 
   \begin{align}
   \inf_{\rho \in \conv(\mathcal{X})} I_c(A\rangle BB',\hat{\mathcal{T}}(\rho^{\otimes k}))
   \geq \inf_{\tau \in P_\eta} I_c(A\rangle BB',\hat{\mathcal{T}}(\tau^{\otimes k})) - \frac{k\delta}{2}
   \label{av_dist_rat_cont}
   \end{align}
   holds. For each $\rho \in \conv(\mathcal{X})$, $\tau \in P_\eta$, we have
   \begin{align}
    \|\rho - \tau\|_1 
    & \leq \|\rho - \mathcal{N}_\eta(\rho) \|_1 + \|\mathcal{N}_\eta(\rho) - \tau \|_1 \\
    & \leq 6\eta + \|\mathcal{N}_\eta(\rho) - \tau \|_1 \label{av_dist_gen_c_b_4}
   \end{align}
    where the last estimation is by (\ref{depol_bound}). From (\ref{av_dist_gen_c_b_4}), 
    we can conclude, that
    \begin{align}
     d_H(\conv(\mathcal{X}),P_\eta)  
      \leq d_H(\mathcal{N}_\eta(\conv(\mathcal{X}), P_\eta) + 6 \eta 
     \leq 7 \eta
    \end{align}
     holds, which implies, via Lemma \ref{coherent_cont_set},
     \begin{align}
      \left|\inf_{\rho \in \conv(\mathcal{X})} I_c(A\rangle BB',\hat{\mathcal{T}}(\rho^{\otimes k}))
       - \inf_{\tau \in P_\eta} I_c(A\rangle BB',\hat{\mathcal{T}}(\tau^{\otimes k}))\right| 
       \leq k \nu(7 \eta) \label{av_dist_cap_cont}.
     \end{align}
     If now $\eta$ is chosen small enough to ensure
     $\nu(7 \eta) < \frac{\delta}{2}$, (\ref{av_dist_rat_cont}), and we conclude,
     collecting inequalities, that the entanglement rate of $\mathcal{D}_l$ is
     \begin{align}
     \frac{1}{l} \log \sr(\phi_l) 
     &\geq \frac{1}{k} \inf_{\tau \in P_\eta} I_c(A\rangle BB',\hat{\mathcal{T}}(\tau^{\otimes k})) - 
     \frac{\delta}{2} \label{av_dist_ach_gen_ende_1} \\
     & \geq \frac{1}{k} \inf_{\rho \in \conv(\mathcal{X})} I_c(A\rangle BB',\hat{\mathcal{T}}
     (\rho^{\otimes k})) - \delta \label{av_dist_ach_gen_ende_2}
    \end{align}
     where (\ref{av_dist_ach_gen_ende_1}) is (\ref{av_dist_ach_gen_pre_rate}), (\ref{av_dist_ach_gen_ende_2})
     is by (\ref{av_dist_rat_cont}). 
  \end{proof}
 \begin{corollary}\label{dist_gen_corr}
  Let $\mathcal{X}$ be a set of states on $\hr_A \otimes \hr_B$. For the AVQS generated by $\mathcal{X}$, it 
  holds
  \begin{align}
   D_{\rightarrow}^{AV}(\mathcal{X}) 
   = D_{\rightarrow}(\conv(\mathcal{X})) 
   = \lim_{l \rightarrow \infty} \frac{1}{k} \sup_{\mathcal{T}\in \Theta_k} \inf_{\tau \in \conv(\mathcal{X})}
      D^{(1)}_{\rightarrow}(\tau^{\otimes k},\mathcal{T}) \label{dist_gen_corr_1}
  \end{align}
  with $D^{(1)}_{\rightarrow}$ being the function defined in (\ref{ent_dist_cap_func_oneshot}), and
  maximization over instruments on $A$'s systems.
 \end{corollary}
 \begin{proof}
  The rightmost equality in (\ref{dist_gen_corr_1}) is Corollary \ref{comp_ent_dist_gen}.1. We 
  prove the first equality. Achievability directly follows from Theorem \ref{av_dist_ach_gen}. 
  For the converse statement, let $\mathcal{X} := \{\rho_s\}_{s \in \bS}$
  and $\sigma \in \conv(\mathcal{X})$. By Carath\'eordory's Theorem (see e.g. Ref. \citenum{webster94}, 
  Theorem 2.2.4.), $\sigma$ can be written as a finite convex combination of elements of $\mathcal{X}$, say  
  \begin{align}
   \sigma = \sum_{s\in\bS'} p(s) \rho_s. 
  \end{align}
  with $|\bS'|<\infty$. Thus, for an $A \rightarrow B$ LOCC channel $\mathcal{D}_l$ for blocklength $l$ with 
  suitable maximally entangled state $\phi_l$, it holds
  \begin{align}
   \inf_{s^l \in \bS^l} F(\mathcal{D}_l(\rho_{s^l}), \phi_l)
   & \leq \min_{s^l \in \bS'^l} F(\mathcal{D}_l(\rho_{s^l}), \phi_l)  \\
   & \leq \sum_{s^l \in \bS'^l} p^l(s^l) F(\mathcal{D}_l(\rho_{s^l}), \phi_l) \\
   & = F(\mathcal{D}_l(\sigma^{\otimes l}), \phi_l). \label{corr_conv}
  \end{align}
  Since ($\ref{corr_conv}$) holds for each element of $\conv(\mathcal{X})$, each rate $R$ 
  which is an $A \rightarrow B$ achievable entanglement distillation rate for the AVQS 
  generated by $\mathcal{X}$ is also achievable for the compound quantum source
  generated by $\conv(\mathcal{X})$, thus the converse statement in Corollary \ref{comp_ent_dist_gen}.1 
  applies. 
 \end{proof} 
  Having determined the one-way entanglement distillation capacity $D_{\rightarrow}^{AV}$, the continuity 
  properties of the capacity function on the r.h.s. of (\ref{av_dist_cap_cont}) imply the following 
  corollary. 
 \begin{corollary}\label{av_dist_cap_cont_corr}
  Identifying each set of states with its closure, $D_{\rightarrow}^{AV}$ is uniformly continuous in the metric 
  defined by the Hausdorff distance on compact sets of density matrices. If $\mathcal{X}, \mathcal{X'} \subset
   \st(\hr_A \otimes \hr_B)$ are two compact sets with $d_H(\mathcal{X}',\mathcal{X}) < \epsilon \leq 
   \frac{1}{2}$ it holds
   \begin{align}
    |D_{\rightarrow}^{AV}(\mathcal{X}') - D_{\rightarrow}^{AV}(\mathcal{X})| \leq \nu(\epsilon).
   \end{align}
 \end{corollary}
  \begin{remark}
   Corollary \ref{av_dist_cap_cont_corr} classifies the AVQS one-way entanglement distillation task as well-behaved in the 
   following sense. Two different AVQS with generating sets being near in the Hausdorff sense will have 
   approximately equal capacities. \newline
   An example for a situation where ``capacity'' is a more fragile quantity is transmission of classical
   messages over an arbitrarily varying quantum channel. The capacity 
   $C_{random}$ for classical message transmission using correlated random codes is continuous, while
   it can be shown, that in some cases, the capacity using deterministic codes, 
   $C_{det}$, is discontinuous on certain points \cite{boche14_a}.
 \end{remark}
 \end{subsection}
 \end{section}
 \begin{section}{On Quantum State Merging for AVQS} \label{av_merging}
 In this Section, we consider quantum state merging in case that the bipartite source $A$ and $B$ have to 
 merge is an AVQS. In Ref. \citenum{bjelakovic13} and Section \ref{comp_merging} of this paper, we have 
 determined the optimal entanglement as well as classical communication cost in case of a compound
 quantum source, and achieved these rates by protocols with merging fidelity going to one exponentially 
 fast (see Section \ref{comp_merging}). Therefore, one would expect, that we can proceed as  
 we did for one-way entanglement distillation for AVQS in Section \ref{av_dist} and 
 build state merging protocols for the AVQS generated by a set 
 $\mathcal{X}$ from suitable protocols for the compound source generated by $\overline{\conv(\mathcal{X})}$. 
 One would expect to be able to prove the equalities 
 \begin{align}
  C_{m,\rightarrow}^{AV}(\mathcal{X}) = C_{m,\rightarrow}(\conv(\mathcal{X})) 
  = \sup_{\rho \in \conv(\mathcal{X})} S(A|B,\rho),  
  \label{av_merging_full}
 \end{align}
 to hold, where again, no difference has to be made between $\conv(\mathcal{X})$ 
 and its closure, because of continuity
 of the conditional von Neumann entropy. Actually, it seems possible, to prove the relation
 \begin{align}
  C^{AV}_{m,\rightarrow}(\mathcal{X}) \leq
  C_{m,\rightarrow}(\conv(\mathcal{X})) \label{av_merging_ach}
 \end{align}
 using Ahlswede's elimination and derandomization techniques (at least if the AVQS is generated by a 
 finite set of states). We do not carry out the argument here. Instead,
 we give a simple counterexample to the relation in (\ref{av_merging_full}). \newline
 Consider a finite set $\hat{\mathcal{X}} := \{\rho_s\}_{s=1}^N$ of bipartite states on a Hilbert space $\hr_A \otimes
 \hr_B$, which is generated by unitaries in the following sense. Let $\rho_1 \in \st(\hr_A \otimes \hr_B)$, where
 we assume $S(A|B,\rho_1) < 0$ and $\dim \hr_A \geq N\cdot \dim \supp(\rho_{A,1})$, $U_1 = \eins_{\hr_A}$
 and $U_2,...,U_N$ unitaries on $\hr_A$ such that with the definitions
 \begin{align}
  \rho_s := U_s\otimes \eins_{\hr_B}(\rho_1)U_s^\ast \otimes \eins_{\hr_B} &&(s \in [N])\label{av_merg_ex_unit}
 \end{align}
 the supports of the $A$-marginals are pairwise orthogonal, i.e.
 \begin{align}
 \supp(\rho_{A,s}) \perp \supp(\rho_{A,s'})  &&(s,s' \in [N], s' \neq s).
 \end{align}
 Note, that our definitions also imply the relations $\rho_{B,s} = \rho_{B,1} \ (s \in [N])$ and
 \begin{align}
  \supp(\rho_s) \perp \supp(\rho_{s'}) &&(s,s' \in [N], s \neq s').
 \end{align}
 In the following we show, that sets constructed in the above described manner are 
 counterexamples to (\ref{av_merging_full}) if $N > 1$.
 \begin{example} \label{av_merging_counterexample}
  For the AVQS generated by $\hat{\mathcal{X}}:= \{\rho_s\}_{s=1}^N$, it holds
  \begin{align}
   C_{m,\rightarrow}^{AV}(\hat{\mathcal{X}}) = C_{m\rightarrow}\left(\conv(\hat{\mathcal{X}})\right) - \log N.
  \end{align}
  The classical $A \rightarrow B$ communication cost for merging of the AVQS $\hat{\mathcal{X}}$
  is upper bounded by 
  \begin{align}
   \sup_{\sigma \in \conv(\hat{\mathcal{X}})} I(A;E,\sigma) -  \log N,
  \end{align}
  where $\rho_p := \sum_{s=1}^N p(s) \rho_s$ for each $p \in \pr([N])$.
 \end{example}
 \begin{proof}[Proof of Example \ref{av_merging_counterexample}]
  Before we prove the claims made in the example, we briefly sketch the argument. 
  Since
  the $A$ marginals are supported on pairwise orthogonal subspaces, $A$ can perfectly detect, given a block
  of $l$ outputs of the AVQS, which of the $s^l \in \bS^l$ is actually realized. In this way, $A$ obtains 
  state knowledge which helps to achieve the desired rates.\newline
  We introduce unitary channels $\mathcal{V}_{A,1},...,\mathcal{V}_{A,N}$ and 
  $\mathcal{V}_{B',1},...,\mathcal{V}_{B',N}$ where we define $\mathcal{V}_{A,s}(\cdot) := 
  U_s (\cdot) U_s^\ast$ with the unitaries from (\ref{av_merg_ex_unit}) and consider $\mathcal{V}_{B',s}$ to 
  be the corresponding unitary channel on the space $\hr_{B'}$ for each $s \in [N]$. 
  For given blocklength $l$, we define unitary channels 
  \begin{align}
   \mathcal{V}_{A,s^l}(\cdot) := \mathcal{V}_{A,s_1} \otimes ...\otimes \mathcal{V}_{A,s_l} 
   \hspace{.3cm} \text{and} \hspace{.3cm}
   \mathcal{V}_{B',s^l}(\cdot):= \mathcal{V}_{B',s_1} \otimes ...\otimes \mathcal{V}_{B',s_l} 
   \label{av_merg_un}
  \end{align}
  for each $s^l = (s_1,...,s_l) \in \bS^l$ accordingly. Thus, the definitions in (\ref{av_merg_ex_unit}) imply
  \begin{align}
   \rho_{s^l} = \mathcal{V}_{A,s^l} \otimes \id_{\hr_B^{\otimes l}}(\rho_1^{\otimes l}). &&(s^l \in [N]^l) 
  \end{align}
  Using the projection $P_s$ onto the support of $\rho_{A,s}$ for each $s \in [N]$, we define a quantum 
  instrument
  \begin{align}
   \hat{\mathcal{A}} := \{\hat{\mathcal{A}}_s\}_{s=1}^N
  \end{align}
  with $\hat{\mathcal{A}}_s(\cdot) := \mathcal{U}_{A,s}\circ P_s(\cdot)P_s^\ast$ for each $s \in [N]$, which implies 
  \begin{align}
   \hat{\mathcal{A}}_{s'}\otimes \id_{\hr_B}(\rho_s) = \delta_{ss'}\rho_1 && (s \in [N]). \label{av_merg_ex_instr}
  \end{align}
  It is known from Ref. \citenum{horodecki07} and Section \ref{comp_merging}, that for each $\delta > 0$ and 
  sufficiently large blocklength $l$, 
  there exists an $(l,k_l,\tilde{D}_l)$ $A \rightarrow B$ merging $\tilde{\mathcal{M}}_l$ such that
  \begin{align}
   F(\tilde{\mathcal{M}}_l\otimes \id_{\hr_E^{\otimes n}}(\psi_1^{\otimes l}), \phi_l \otimes \psi_1'^{\otimes l})
   \geq 1 - 2^{-lc} \label{av_merg_ex_fid_1}
  \end{align}
  holds with a constant $c > 0$, where $\psi_1$ is a purification of $\rho_1$ and $\phi_l$ a maximally 
  entangled state shared by $A$ and $B$ with 
  \begin{align}
   - \frac{1}{l}\log \sr(\phi_l) \leq S(A|B,\rho_1) + \delta \label{av_merg_ex_qrate_0}
  \end{align}
  and where for the classical communication rate 
  \begin{align}
   \frac{1}{l} \log \tilde{D}_l \leq I(A;E,\rho_1) + \delta \label{av_merg_ex_crate_0}
  \end{align}
  holds. We combine the instrument $\hat{\mathcal{A}}$ and the unitary channels from (\ref{av_merg_un}) with 
  $\tilde{\mathcal{M}}_l$ to build a merging LOCC $\mathcal{M}_l$ suitable for merging the AVQS generated 
  by $\hat{\mathcal{X}}$ and define 
  \begin{align}
   \mathcal{M}_l
   &:= \sum_{s^l \in [N]^l} (\mathcal{V}_{B',s^l} \otimes \id_{\hr_B^{\otimes l}}) 
   \circ \tilde{\mathcal{M}_l} \circ(\hat{\mathcal{A}}_{s^l} \otimes \id_{\hr_B^{\otimes l}}). 
  \end{align}
  Clearly, $\mathcal{M}_l$ is an $A \rightarrow B$ LOCC channel. Explicitly, 
  inspection of the above definition shows, 
  that $\mathcal{M}_l$ is an $(l,k_l,D_l)$ $A\rightarrow B$ merging where one of
  \begin{align}
   D_l = \tilde{D}_l \cdot N^l \label{avqs_merging_ex_class_rate}
  \end{align}
   different classical messages has to be communicated within action of $\mathcal{M}_l$. 
   Moreover, for each $s^l \in [N]^l$, it holds
  \begin{align}
   &F(\mathcal{M}_l \otimes \id_{\hr_{E}^{\otimes l}}(\psi_{s^l}),\phi_l \otimes \psi'_{s^l}) \\
   &\overset{(a)}{=} \sum_{m^l \in [N]^l} F((\mathcal{V}_{B',m^l} \otimes \id_{\hr_B^{\otimes l}})
   \circ \tilde{\mathcal{M}}_l \circ(\hat{\mathcal{A}}_{m^l} \otimes \id_{\hr_{B}^{\otimes l}}) 
   \otimes \id_{\hr_E^{\otimes l}}(\psi_{s^l}), \phi_l \otimes \psi'_{s^l})  \\
   &\overset{(b)}{=} F(\tilde{\mathcal{M}}_l \otimes \id_{\hr_E^{\otimes l}} (\psi_{1}^{\otimes l}), 
   \phi_l \otimes (\mathcal{V}_{B',s^l}^\ast \otimes \id_{\hr_{BE}^{\otimes l}})(\psi'_{s^l}))  \\
   &\overset{(c)}{=} F(\tilde{\mathcal{M}}_l \otimes \id_{\hr_E^{\otimes l}}
   (\psi_{1}^{\otimes l}), \phi_l \otimes (\psi'_1)^{\otimes l}) \\
   &\overset{(d)}{\geq} 1 - 2^{-lc}, \label{av_merg_ex_fid_2}
   \end{align}
   where (a) is the definition of $\mathcal{M}_l$ plus linearity of the fidelity in the first argument in 
   the present situation, (b) is because
   \begin{align}
    \hat{\mathcal{A}}_{m^l} \otimes \id_{\hr_{BE}^{\otimes l}}(\psi_{s^l}) = \delta_{m^ls^l} \psi_1^{\otimes l}
   \end{align}
   holds implied by (\ref{av_merg_ex_instr}) together with the fact, that the fidelity is invariant under 
   action of unitary channels applied simultaneously on both arguments. Equality (c) follows from 
   (\ref{av_merg_ex_unit}), and (d) is by (\ref{av_merg_ex_fid_1}). It remains to evaluate the rates. 
   It is well known, that for each ensemble $\{q(x), \rho_x \}_{x\in \mathbf{X}}$ of quantum states having 
   pairwise orthogonal supports, it holds
   \begin{align}
   S\left(\sum_{x \in \mathbf{X}} q(x) \rho_x\right) = \sum_{s \in \mathbf{X}} q(x) S(\rho_x) + H(q).
   \end{align}
   where $H(q)$ is the Shannon entropy of $q$. Thus, for each
   $p \in \pr([N])$, $\rho_p := \sum_{s \in [N]} p(s) \rho_s$ we yield
   \begin{align}
    S(A|B,\rho_p) 
    &= S(\rho_p) - S(\rho_{B,p}) \\
    &= \sum_{s \in [N]} p(s) S(\rho_s) + H(p) - S(\rho_{B,1})\\
    &= S(A|B,\rho_1) + H(p)
   \end{align}
    and 
   \begin{align}
    I(A;E,\rho_p) 
    &= S(\rho_{A,p}) + S(A|B,\rho_p) \\
    &= \sum_{s \in [N]} p(s) S(\rho_{A,s}) + S(A|B,\rho_1) + 2 H(p)\\
    &= I(A;E,\rho_1) + 2 H(p).
   \end{align}
   Taking maxima over all $p \in \pr([N])$ and rearranging equations, we arrive at
   \begin{align}
    S(A|B,\rho_1) = \max_{p \in \pr([N])} S(A|B,\rho_p) - \log N \label{av_merg_ex_qrate}
   \end{align}
   and
   \begin{align}
    I(A;E,\rho_1) = \max_{p \in \pr([N])} I(A;E,\rho_p) - 2 \log N. \label{av_merg_ex_crate}
   \end{align}
   Note, that 
   \begin{align}
    C_{m,\rightarrow}(\conv(\hat{\mathcal{X}})) = \max_{p \in \pr([N])} S(A|B,\rho_p) \label{comp_rate_formula}
   \end{align}
   by Proposition \ref{comp_merging_complete}.
   Combining (\ref{av_merg_ex_qrate}) with (\ref{av_merg_ex_qrate_0}) and (\ref{comp_rate_formula}) together 
   with (\ref{av_merg_ex_fid_2}) shows, that 
   \begin{align}
    C_m^{AV}(\mathcal{X}) \leq C_m(\conv(\hat{\mathcal{X}})) - \log N + \delta
   \end{align}
   holds. The converse is valid by the merging cost converse for single states
   \cite{horodecki07}. Moreover, by 
   (\ref{av_merg_ex_crate}), our protocols have classical $A \rightarrow B$ classical communication rates  
   with
   \begin{align}
       \limsup_{l \rightarrow \infty} \frac{1}{l}\log D_l
        &= \limsup_{l \rightarrow \infty} \frac{1}{l}\log(\tilde{D}_l\cdot N^l) \label{avqs_merging_ex_ende_1}\\
        &\leq I(A;E,\rho_1) + \delta + \log N \label{avqs_merging_ex_ende_2} \\
        &= \max_{p \in \pr([N])} I(A;E,\rho_p) - \log N + \delta \label{avqs_merging_ex_ende_3}
   \end{align}
   where (\ref{avqs_merging_ex_ende_1}) follows from (\ref{avqs_merging_ex_class_rate}), 
   (\ref{avqs_merging_ex_ende_2}) is by (\ref{av_merg_ex_crate_0}), and (\ref{avqs_merging_ex_ende_3})
   is by (\ref{av_merg_ex_crate}). Since $\delta > 0$ was an arbitrary positive number, we are done. 
   \end{proof}  
\end{section}
 \begin{section}{Conclusion} \label{conclusion}
  In this work, we have shown simultaneous achievability of the optimal entanglement as well as 
  classical communication cost of one-way quantum state merging in case, that the source to merge 
  is a compound quantum source. In this way, we completed our work on 
  quantum state merging for compound sources begun in Ref. \citenum{bjelakovic13}. \newline
  We also determined the optimal entanglement rates for one-way entanglement distillation in case, that 
  the source from which the entanglement is distilled is an AVQS. In this case, Ahlswede's 
  robustification and elimination technique turned out to be appropriate tools, 
  and we can in fact, by the elimination technique, achieve each rate below the entanglement capacity
  with fidelity going to one exponentially fast and simultaneously approximate the classical communication
  rate of the utilized protocols for the compound source generated by the convex hull of the AVQS generating 
  set.
  \newline
  Imposing a simple example of a class of AVQS, we demonstrated, that applying the robustification 
  and elimination technique to suitable protocols for the corresponding compound source (generated by
  the convex hull of the AVQS generating set), is insufficient in general.\newline
  Another situation, where the above standard approach is not suitable, is the problem of proving 
  achievability of the strong subspace capacity of an arbitrarily varying quantum channel (AVQC). 
  In this case, the problem is not immediately accessible for the 
  robustification technique, and this deficiency was overcome in Ref. \citenum{ahlswede13} by first determining the
  capacity of the 
  AVQC for entanglement transmission, and then showing equality of the capacities utilizing fairly nontrivial
  results from convex high-dimensional convex geometry. \newline
  The quantum state merging problem for AVQS, in contrast, seems accessible to robustification and elimination.
  However, application leads to suboptimal rates in some cases, as Example \ref{av_merging_counterexample} 
  shows.\newline
  In fact, a closer look to Example \ref{av_merging_counterexample} reveals, that the achievability result 
  asserted by the inequality in (\ref{av_merging_ach}), is not only suboptimal, but also meaningless in a 
  qualitative sense for some AVQS.\newline
  Imagine a situation, in which the communication parties do not have 
  any access to shared pure entanglement resources and they want to merge the AVQS generated by a set 
  $\hat{\mathcal{X}}$ as in Example \ref{av_merging_counterexample} with $C_m\left
  (\conv(\hat{\mathcal{X}})\right) > 0$, where the number $N$ of states in 
  $\hat{\mathcal{X}}$ is assumed to be bounded 
  \begin{align}
   N > \exp(C_m(\conv(\hat{\mathcal{X}})).
  \end{align}
  Having only protocols according to the achievability result (\ref{av_merging_ach}) at hand, they infer, 
  that merging is impossible in their situation, while Example \ref{av_merging_counterexample} shows, 
  that merging of the AVQS is, in fact, possible without external entanglement resources. \newline
  Summarizing our considerations, we notice with some regret, that in case of quantum state merging, the 
  merging cost of an AVQS generated by a set $\mathcal{X}$ seems, at least not immediately, related to the 
  merging cost of the corresponding compound source generated by $\conv(\mathcal{X})$. A merging cost function
  presumably will, involve LOCC pre- and post-processing 
  maximization. Probably, the merging cost for AVQS will require a multi-letter characterization.
  \end{section}
 \begin{section}*{Acknowledgements}
  We express our deep gratitude to Igor Bjelakovi\'{c}, our former colleague and constant source of inspiration
  for years. Having prepared Ref. \citenum{bjelakovic13} with us, he also 
  stimulated the initial discussions leading to the present paper. Partial results of this paper where presented 
  in a talk at the 2014 IEEE International Symposium on Information Theory (ISIT 2014) 
  without detailed proofs\cite{boche14_b}.\newline
  The work of H.B. is supported by the DFG via grant BO 1734/20-1 and by the BMBF via grant 01BQ1050.
  G.J. gratefully acknowledges the support of the TUM Graduate School / Faculty Graduate Center FGC-EI at
  Technische Universit\"at M\"unchen, Germany.
 \end{section}
 \begin{section}{Appendix}
 \begin{subsection}{Proof of the bound in Eq. (\ref{entropy_estimation_constant_bound})}
 \label{entropy_estimation_appendix}
  Let $\eta > 0$ be fixed and $p,q$ probability distributions on $[d]$, such that 
  \begin{align}
   |H(p)-H(q)| \geq \eta \label{appendix1_1}
  \end{align}
  holds. It is well known, that the Shannon entropy is uniformly continuous in the variation distance 
  (see e.g. \cite{csiszar11}), 
  it holds
  \begin{align}
   |H(p)-H(q)| \leq f(\|p-q\|_1) \label{appendix1_2}
  \end{align}
  with a strictly monotonically increasing function $f$. Therefore, (\ref{appendix1_1}) and (\ref{appendix1_2}) 
  lead to 
  \begin{align}
   0 < 2c_3 := \frac{1}{2\ln 2}f^{-1}(\eta)^2 \leq  \frac{1}{2 \ln 2}\|p-q\|_1^2 \leq D(p||q),
  \end{align}
  where the rightmost inequality is Pinsker's inequality $D(p||q) \geq \frac{1}{2\ln 2}\|p-q\|_1^2$. 
  Since $p$ and $q$ where arbitrary probability distributions
  on $[d]$ with entropy distance bounded below by $\eta$, for each $i \in [N]$, the bound in 
  (\ref{entropy_estimation_constant_bound}) is valid for each $i \in [N]$.
\end{subsection} 

\end{section}

\end{document}